\documentclass[11pt,letterpaper]{article}

\newcommand{\ifabs}[2]{#2}

\newcommand{\ifarxiv}[2]{#1}

\usepackage{comment}
\usepackage{cite}
\usepackage{fullpage}
\usepackage{color}
\usepackage{epic}
\usepackage{eepic}
\usepackage{epsfig}
\usepackage{xspace}
\usepackage{graphicx}
\usepackage{amsmath}
\usepackage{amssymb,amsthm}
\usepackage[ruled,vlined]{algorithm2e}
\ifabs{
\usepackage{times}
}{}

\ifabs{ \excludecomment{fullproof}

\includecomment{abs}
\excludecomment{fullpaper}

}{

\excludecomment{proofsketch} \excludecomment{abs}
\includecomment{fullpaper}
}
 
\excludecomment{suppress}

\newcommand{\todo}[1]{\typeout{TODO: \the\inputlineno: #1}\textbf{[[[ #1 ]]]}}

\newcommand{\concept}[1]{\emph{#1}}

\newtheorem{theorem}{Theorem}[section]
\newtheorem{lemma}[theorem]{Lemma}
\newtheorem{corollary}[theorem]{Corollary}
\newtheorem{definition}[theorem]{Definition}

\newtheorem{claim}[theorem]{Claim}
\newtheorem*{claim*}{Claim}
\newtheorem{proposition}[theorem]{Proposition}

\newcommand{\newloglike}[2]{\newcommand{#1}{\mathop{\rm #2}\nolimits}}
\newloglike{\sgn}{sgn}

\newcommand{\nul}[1]{{\it et al.\/}}

\newcommand{\Pin}[3]{\mbox{\textsc{Pin}}_{{#1}}^{{#2}}\left({#3}\right)}
\newcommand{\Peer}[3]{\text{\textsc{Peer}}_{{#1}}^{{#2}}\left({#3}\right)}

\newcommand{\regular}[1]{${#1}$-regular}
\newcommand{\holant}{\mathrm{hol}}
\newcommand{\Holant}{\mathrm{Holant}}

\newcommand{\tw}{\mathrm{tw}}

\begin{document}

\title{Approximate Counting via Correlation Decay on Planar Graphs}
\author{
Yitong Yin\thanks{Supported by the National Science Foundation of China under Grant No. 61003023
and No. 61021062.}\\
Nanjing University\\
\texttt{yinyt@nju.edu.cn}
\and
Chihao Zhang\\
Shanghai Jiaotong University\\
\texttt{chihao.zhang@gmail.com}
}
\date{}

\maketitle

\begin{abstract}


We show for a broad class of counting problems, correlation decay (strong spatial mixing) implies FPTAS on planar graphs. The framework for the counting problems considered by us is the Holant problems with arbitrary constant-size domain and symmetric constraint functions. We define a notion of regularity on the constraint functions, which covers a wide range of natural and important counting problems, including all multi-state spin systems, counting graph homomorphisms, counting weighted matchings or perfect matchings, the subgraphs world problem transformed from the ferromagnetic Ising model, and all counting CSPs and Holant problems with symmetric constraint functions of constant arity.

The core of our algorithm is a fixed-parameter tractable algorithm which computes the exact values of the Holant problems with regular constraint functions on graphs of bounded treewidth. By utilizing the locally tree-like property of apex-minor-free families of graphs, the parameterized exact algorithm implies an FPTAS for the Holant problem on these graph families whenever the Gibbs measure defined by the problem exhibits strong spatial mixing. We further extend  the recursive coupling technique to Holant problems and establish strong spatial mixing for the ferromagnetic Potts model and the subgraphs world problem. As consequences, we have new deterministic approximation algorithms on planar graphs and all apex-minor-free graphs for several counting problems.
\end{abstract}

\ifabs{ \setcounter{page}{0} \thispagestyle{empty} \vfill
\pagebreak }{
}

\section{Introduction}
In study of counting algorithms, many counting problems can be formulated as computing the \concept{partition function}: 
\[
Z(G(V,E))
=\sum_{\sigma\in[q]^V}\prod_{uv\in E}\Phi_{E}(\sigma(u),\sigma(v))\prod_{v\in V}\Phi_{V}(\sigma(v)),
\]
where $\Phi_{E}:[q]^2\to\mathbb{C}$ and $\Phi_V:[q]\to\mathbb{C}$ are symmetric functions. 
This model is called \concept{spin system} in Statistical Physics. It has vertices as variables and edges as constraints, and the partition function returns the total weight of all \concept{configurations}. 
Many natural combinatorial problems such as counting independent sets, $q$-colorings, or graph homomorphisms can be expressed in this way.

We consider a framework that 
encompasses a much broader class of counting problems, namely, the \concept{Holant problems}.

An instance of a Holant problem is an $\Omega=(G(V,E),\{f_v\}_{v\in V})$, where $G$ is a graph, and each $f_v$ is a function that maps tuples in $[q]^{\deg(v)}$ to function values. The \concept{Holant} of $\Omega$ is defined as
\[
\holant(\Omega)=\sum_{\sigma\in[q]^E}\prod_{v\in V}f_v\left(\sigma\mid_{E(v)}\right),
\]
where $f_v(\sigma\mid_{E(v)})$ evaluates $f_v$ on the restriction of $\sigma$ on incident edges $E(v)$ of vertex $v$. The Holant problem $\Holant(\mathcal{G},\mathcal{F})$ specified by a graph family $\mathcal{G}$ and a function family $\mathcal{F}$, 
is the problem of computing $\holant(\Omega)$ for all valid instances $\Omega$ defined by graphs from $\mathcal{G}$ and functions from $\mathcal{F}$.

The term Holant is coined by Valiant in \cite{valiant2008holographic} in studying of holographic algorithms. The formal framework of Holant problems is proposed in \cite{cai2009holant} by Cai, Lu and Xia.

The Holant framework is extremely expressive. Using the bipartite incidence graph to represent the participants of variables in constraints and choosing appropriate functions at vertices on both sides, computing the partition functions of spin systems and more generally counting CSPs can all be represented as special classes of Holant problems.

An algorithmic significance of Holant problems is that they are outcomes of holographic transformations. 
The holographic algorithms proposed by Valiant~\cite{valiant2006accidental,valiant2008holographic} compute exact solutions to the counting problems on planar graphs by transforming to problems solvable by the FKT algorithm~\cite{fisher1961statistical,kasteleyn1961statistics,temperley1961dimer} for 
counting planar perfect matchings.
In the realm of approximate counting, perhaps the most successful (implicit) using of holographic transformation and Holant problem is the FPRAS for ferromagnetic Ising model given by Jerrum and Sinclair in their seminal work~\cite{jerrum1993polynomial}. 
The transformation in~\cite{jerrum1993polynomial} from the spins world to the subgraphs world is indeed a holographic transformation, and the resulting subgraphs world  problem is a Holant problem.\footnote{This actually happened more than a decade earlier than the concepts of  holographic algorithm and Holant problem formally defined.} In these examples, the original counting problem is transformed to a Holant problem which has efficient exact or approximate algorithms. Therefore, the following problem is fundamental to the study of counting algorithms: 

Characterize the tractability of \emph{exact} computation and \emph{approximation} of $\Holant(\mathcal{G},\mathcal{F})$ in terms of graph family $\mathcal{G}$ and function family $\mathcal{F}$.

\ifabs{\vspace{-12pt}}{}
\paragraph{Exact computation.}
The exact computation of Holant problems has been well studied on general graphs \cite{caiholant,cai2012complete,cai2011computational,cai2011dichotomy,cai2008holographic,cai2011holo,luhuang12} and planar graphs \cite{cai2011holographic,cai2010holographic,valiant2008holographic}, sometimes in form of dichotomy theorems, which states that every problem in the considered framework is either \#P-hard or having polynomial-time algorithm. 
Very recently, a dichotomy theorem \cite{cai2012complete} is proved for Holant problems with complex-valued functions on general graphs, concluding a long series of dichotomies on Holant problems.
All these results consider Holant problems with boolean domain ($q=2$). Meanwhile, some special classes of Holant problems are more thoroughly understood, such as counting graph homomorphisms or counting CSP. 
For these specialized frameworks, dichotomy theorems were proved in a very general setting with complex-valued functions on general-sized domains~\cite{cai2010graph,Cai:2012:CCC:2213977.2214059}. See \cite{chen2011guest} for a good survey on these subjects.

Speaking very vaguely, the dichotomy theorems tell us that except for some rare cases almost all Holant problems are hard. 
Then a problem of algorithmic significance is to establish tractable results for $\Holant(\mathcal{G},\mathcal{F})$ on more \emph{refined} graph families $\mathcal{G}$, e.g.~graphs with fixed parameters or forbidden minors, especially for general domain size $q>2$.

\ifabs{\vspace{-12pt}}{}
\paragraph{Approximation.}
We focus on deterministic approximate counting algorithms, specifically, the deterministic fully polynomial time approximation scheme (FPTAS). 
A central topic in this direction is the relation between correlation decay (strong spatial mixing) and approximability of counting. 


Correlation decay is a property about the marginal distribution, which is computationally equivalent to counting by the famous self-reduction of Jerrum-Valiant-Vazirani~\cite{jerrum1986random}. The correlation decay property says that faraway vertices have little influence on the marginal distribution of local states, thus marginal probabilities should be well-approximated by local information only. However, as noted in~\cite{gamarnik2007correlation,bayati2007simple}, this sufficiency of \emph{local information} does not immediately yield efficient \emph{local computation}. 
Two tools are invented to bridge this gap: the self-avoiding-walk tree (SAW-tree) of Weitz~\cite{weitz2006counting} and the computation tree of Gamarnik and Katz~\cite{gamarnik2007correlation}. Both transform the original graph to a tree structure in which the marginal probabilities can be efficiently computed by recursions.  With the SAW-tree the implication from strong spatial mixing to FPTAS is proved for 2-state spin systems~\cite{weitz2006counting}, which becomes a foundation for several important algorithmic results~\cite{li2012approximate,li2011correlation,sinclair2012approximation,restrepo2011improved}. It is also proved in a long series of beautiful work~\cite{galanis2011improved,galanis2012inapproximability,sly2010computational,sly2012computational,inapp_MWW09,IS_DFJ02} that for the same class of counting problems lack of correlation decay implies inapproximability.

\ifabs{
The relation between correlation decay and approximability for broader classes of counting problems is widely open.
}
{
The relation between correlation decay and approximability for broader classes of counting problems is widely open, 
because of following technical challenges:




\begin{itemize}
\item 
It is known~\cite{sly2008uniqueness} that for domain size $q>2$ tree may not always represent the extremal case for correlation decay. Thus in order to use correlation decay to support approximate counting for those problems, the local computation has to be done on structures other than trees.






\item 
Even on trees, the current recursion-based computation critically relies on the simplicity of constraint functions, as in the cases of spin systems and matchings. For general Holant problems, even on trees and when $q=2$, it is not known whether simple recursion exists.
\end{itemize}
}

\subsection{Our results}
We make progress on both exact and approximate computation of Holant problems by establishing connections between them.

We characterize a broad class of Holant problems whose exact computation is tractable on tree-like graphs and FPTAS is implied by strong spatial mixing on planar graphs. These Holant problems are characterized by a notion of regularity introduced by us on the constraint functions. Intuitively, being regular as a function means that the entropy of any input or partial input is constant. This covers a large family of important counting problems, including all spin systems,  graph homomorphisms, counting CSP with symmetric constraints of bounded arity, matchings, perfect matchings, the subgraphs world problem in~\cite{jerrum1993polynomial}, etc.

For this broad class of counting problems, we give a fixed-parameter tractable algorithm which computes the exact value of counting in time $2^{O(k)}\cdot\mathrm{poly}(n)$ on graphs of size $n$ and treewidth $k$. Based on this parameterized algorithm, strong spatial mixing implies FPTAS on apex-minor-free graphs, which include planar graphs as special case.

We also apply the recursive coupling technique of Goldberg \emph{et al.}~\cite{goldberg2005strong} to analyze the strong spatial mixing for Holant problems. As examples, we have deterministic FPTAS on planar graphs, and more generally on all apex-minor-free graphs for the following counting problems:
\begin{itemize}
\item Counting $q$-colorings on triangle-free planar graphs of maximum degree $\Delta$ when $q>\alpha\Delta-\gamma$ where $\alpha\approx1.76322$ 
and $\gamma\approx0.47031$. This is just directly applying~\cite{goldberg2005strong}.
\item The subgraphs world with parameter $\mu,\lambda<1$ on planar graphs of maximum degree $\Delta$ when $\Delta<\frac{(1+\lambda\mu^2)^2}{1-\mu^2}$, and as a consequence the ferromagnetic Ising model\footnote{Due to the classic results of 1960s in Statistical Physics~\cite{fisher1966dimer,kasteleyn1963dimer} and the recent theory of holographic algorithms~\cite{cai2011holographic}, the planar Ising model with zero field $B=0$ is solvable exactly in polynomial time. Here we consider Ising models with general field $B$, which is \#P-hard even on planar graphs proved implicitly in~\cite{kowalczyk2010dichotomy}.} with inverse temperature $\beta$ and external field $B$ when $\Delta<\frac{(e^{2\beta+4B}+e^{2\beta}+2e^{2B})^2}{e^{2B}(e^{2\beta}+1)^2(e^{2B}+1)^2}$.
\item Ferromagnetic $q$-state Potts model of inverse temperature $\beta$ on planar graphs of maximum degree $\Delta$ when $\beta<\frac{\ln\left(\frac{q-2}{\Delta-1}\right)}{\Delta+1}$, which vastly improves the mixing condition in~\cite{gamarnik2007correlation} for FPTAS on general graphs and is close to the $\beta=O(\frac{1}{\Delta})$ bound conjectured in~\cite{gamarnik2007correlation}.
\end{itemize}

\ifabs{\vspace{-12pt}}{}
\paragraph{Technical contributions.}
Our parameterized algorithm does not directly use the tree decomposition. Instead, we define a new decomposition called the separator decomposition, which recursively separates the graph by small graph separators into components of limit-sized boundaries. This is quite different from the known treewidth-based approaches for spin systems, e.g.~the \concept{junction tree} algorithm; and this new construction more closely aligns with the \emph{conditional independence} property: 
conditioning on any fixed assignment on a separator, the states of separated vertices are independent. 
The construction of separator decomposition makes explicit connections between the separable structure of tree-like graphs and the conditionally independent nature of counting problems defined by local constraints. As a result, our algorithm can deal with much broader class of counting problems other than just spin systems.

Unlike previous approximation algorithms via correlation decay, where the decay is verified on a tree of size exponential in the size of original graph, our FPTAS only relies on the correlation decay on the original graph. Thus we can directly apply those ``decay-only'' results such as~\cite{goldberg2005strong} to get FPTAS.  Since we do not explode the size of the graph, the FPTAS can even be supported by \concept{single-site} correlation decays.

\ifabs{}{
\subsection{Related work}
The use of correlation decay technique for designing FPTAS for counting problems was initiated in \cite{bandyopadhyay2008counting,weitz2006counting} and has been successfully applied to many problems~\cite{bayati2007simple,gamarnik2007correlation,galanis2011improved}, especially for computing the partition function of Ising model~\cite{li2012approximate,li2011correlation,sinclair2012approximation}. The technique of recursive coupling has been used to prove the property of correlation decay~\cite{martinelli2003ising,martinelli2007fast,goldberg2005strong,goldberg2006improved}. 

The locally tree-like property of planar graphs and apex-minor-free graphs provides structure information to develop both exact and approximation algorithms on decision and optimization problems, some examples include~\cite{baker1994approximation,eppstein1995subgraph,demaine2009approximation,flum2001fixed}. 

A framework for parameterized complexity of counting problems was proposed in \cite{arvind2002approximation,flum2004parameterized,mccartin2002parameterized}. 
The parameterized complexity of computing partition functions has been studied via probabilistic inference in graphical model~\cite{chandrasekaran2008complexity}.
Some logical approaches have also been extended to counting problems on structures with small treewidth or local treewidth~\cite{arnborg1991easy,courcelle2001fixed,frick2004generalized}.
}

\section{Models and statement of results}

\subsection{Holant problems}
Let $[q]=\{0,1,\ldots,q-1\}$ be a \concept{domain} of size $q$, where $q\ge 2$ is an finite integer. Let $f:[q]^d\rightarrow\mathbb{F}$ be a $d$-ary function where $\mathbb{F}$ is a field. In this paper, we consider either $\mathbb{F}=\mathbb{C}$ the complexes or $\mathbb{F}=\mathbb{R}^+$ the nonnegative reals. To avoid issues of computation model, we assume all number are algebraic.
We allow the function arity $d$ to be 0. When $d=0$, the only member of $[q]^0$ is the empty tuple $\xi$, and a 0-ary function $f$ maps $\xi$ to a function value. We call such function $f$ a \concept{trivial function}.

A $d$-ary function is \concept{symmetric} if $f(x_1,\ldots,x_{d})=f(x_{\rho(1)},\ldots,x_{\rho(d)})$ for any permutation $\rho$ of $\{1,2,\ldots,d\}$. When $q=2$, functions have boolean domain and a $d$-ary symmetric function $f$ can be denoted by $[f_0,f_1,\ldots,f_d]$ where $f_k$ specifies the function value for the input tuple with Hamming weight $k$. For example the \textsc{Equality} function is denoted as $[1,0,0\ldots,0,1]$.

\ifabs{}{
Let $\Phi_E:[q]^2\to\mathbb{C}$ and $\Phi_V:[q]\to\mathbb{C}$ be two symmetric functions. The \concept{partition function} of an undirected graph $G(V,E)$ is defined as
\begin{align*}
Z(G)
&=\sum_{\sigma\in[q]^V}\prod_{\{u,v\}\in E}\Phi_E(\sigma(u),\sigma(v))\prod_{v\in V}\Phi_V(\sigma(v)).
\end{align*}
This is called a \concept{$q$-state spin system}. 
}
Let $\Omega=(G(V,E),\{f_v\}_{v\in V})$ be an instance, where each $f_v$, called a \concept{constraint function}\ifabs{,}{ or a \concept{signature},} is a $d$-ary symmetric function with $d=\deg(v)$.
We define the \concept{Holant} of $\Omega$ as  
\ifabs{$\holant(\Omega)=\sum_{\sigma\in[q]^E}\prod_{v\in V}f_v\left(\sigma\mid_{E(v)}\right)$,}
{\begin{align*}
\holant(\Omega)=\sum_{\sigma\in[q]^E}\prod_{v\in V}f_v(\sigma\mid_{E(v)}),
\end{align*}}
where $f_v(\sigma\mid_{E(v)})$ evaluates $f_v$ on the restriction of $\sigma$ on incident edge of $v$.


Let $\mathcal{G}$ be a family of graphs and $\mathcal{F}$ be a family of functions. A \concept{Holant problem} $\Holant(\mathcal{G},\mathcal{F})$ is a computation problem that given as input an instance $\Omega=(G(V,E), \{f_v\}_{v\in V})$ where $G\in\mathcal{G}$ and all $f_v$ are from $\mathcal{F}$, compute $\holant(\Omega)$.


\ifabs{}{
A symmetric function can be represented by a vector enumerating the function values for all inputs (up to symmetry). 
The number of symmetry classes of $\sigma\in[q]^d$ 
equals the number of weak $q$-composition of integer $d$, which is ${d+q-1\choose q-1}$. Thus symmetric $d$-ary functions can be represented by vectors of length polynomial in $d$.

Spin systems can be represented as special class of Holant problems. For a graph $G(V,E)$, let $\mathcal{I}_G$ denote the \concept{incidence graph} of $G$, i.e.~$\mathcal{I}_G=(V_1,V_2,E')$ is a bipartite graph with $V_1=V$, $V_2=E$ and $(v,e)\in E'$ if and only if edge $e$ is indecent to vertex $v$ in $G$. For a spin system defined by functions $\Phi_E$ and $\Phi_V$ on a graph $G$, we can transform it to a Holant instance $\Omega=(\mathcal{I}_G,\{f_v\}_{v\in V\cup E})$, 
where $f_v=\Phi_E$ for right vertices $v\in E$ and for left vertices $v\in V$, $f_v$ is the generalized \textsc{Equality} function defined as $f(x_1,\ldots,x_d)=\Phi_V(x_1)$ if $x_1=\cdots=x_d$ and $f(x_1,\ldots,x_d)=0$ if otherwise. 
It is easy to check that $\holant(\Omega)=Z(G)$.
}

\subsection{Regular functions}
Our characterization of Holant problems relies on a ``pinning'' operation on  symmetric functions. 
The pinning operation on a function defines a new function with smaller arity by by fixing (pinning) the values of some of the variables.
\begin{definition}[pinning]
Let $f:[q]^d\rightarrow\mathbb{F}$ be a $d$-ary symmetric function. Let $0\le k\le d$ and $\tau\in[q]^k$.  
We define that $\Pin{\tau}{}{f}=g$ where $g: [q]^{d-k}\rightarrow\mathbb{F}$ is a $(d-k)$-ary symmetric function such that 
\[
\forall\sigma\in[q]^{d-k},\quad g(\sigma)=f(\sigma(1),\ldots, \sigma(d-k),\tau(1),\ldots,\tau(k)).
\]
Specifically, when $k=0$ the resulting function $g=f$; and when $k=d$, the resulting function $g$ is a trivial function $f(\sigma)$.
\end{definition}
Note that since $f$ is symmetric, the positions of $\tau(1),\ldots,\tau(k)$ in $f$ does not matter, and the pinning of a symmetric function is still symmetric.

To exemplify the effect of pinning, consider the case when $q=2$ and a function $f$ is represented in form $[f_0,f_1,\ldots,f_d]$. 
For a $\sigma\in[2]^{k}$ that $\sigma$ has $\ell$ many 1s, we have $\Pin{\sigma}{}{f}=[f_\ell, f_{\ell+1}, \ldots, f_{d-(k-\ell)}]$. That is, the $\Pin{\sigma}{}{f}$ for a $\sigma\in[2]^k$ returns a ``sliding windows'' of length $d-k$ in $[f_0,f_1,\ldots,f_d]$ whose starting position is determined by the number of 1s in $\sigma$.


A notion of regularity of symmetric functions can be defined by limiting the outcomes of pinning.
\begin{definition}[constant regularity]
A symmetric function $f:[q]^d\rightarrow\mathbb{F}$ is called \concept{\regular{C}} if for all $0\le k\le d$, it holds that 
\[
\left|\left\{\Pin{\tau}{}{f} \mid \forall \tau\in[q]^k\right\}\right|\le C.
\]
A family $\mathcal{F}$ of symmetric functions is called \concept{regular} if there exists a finite constant $C>0$ such that every $f\in\mathcal{F}$ is \regular{C}.
\end{definition}

\ifabs{
Holant problems with regular constraint functions covers all $q$-state spin systems,  (weighted) matchings, perfect matchings, and the subgraphs world transformed from the Ising model~\cite{jerrum1993polynomial}, and all counting CSPs and Holant problems with symmetric constraints of constant arity. The explanation and a formal condition for being regular are presented in Appendix.
}
{
We then give some sufficient conditions for regular symmetric functions.

\begin{proposition}
Let $f:[q]^d\rightarrow\mathbb{F}$ be a symmetric function. For $\sigma\in[q]^d$ and $i\in[q]$, let $n_i(\sigma)=|\{1\le j\le d\mid \sigma(j)=i\}|$ be the number of $i$-entries in $\sigma$.   
\begin{itemize}
\item (bounded arity) $f$ is \regular{{d+q-1\choose q-1}}. 
\item (cyclic) If there is  a $c>0$ such that $f(\sigma)$ depends only on $(n_1(\sigma)\bmod c,\ldots,n_q(\sigma)\bmod c)$ then $f$ is \regular{c^{q-1}}.
\item (constant exceptions) If there is a \regular{C} $g:[q]^d\rightarrow\mathbb{F}$ and a $c\ge 0$ such that $f$ and $g$ differ only at those $\sigma\in[q]^d$ that $n_i(\sigma)\ge d-c$ for some $i\in[q]$, then $f$ is \regular{\left(C+q\cdot{c+q-1\choose q-1}\right)}.
\end{itemize}
\end{proposition}
Therefore all constant-ary symmetric functions, \textsc{Equality} and generalized \textsc{Equality}, and the \textsc{Not-All-Equal} are all regular. This covers all $q$-state spin systems. 

For boolean domain, a function $[f_0,f_1,\ldots,f_d]$ is regular either if it is cyclic, i.e.~$f_k=\lambda_{k\bmod c}$ for some constant $c$, 
or if it becomes cyclic after removing constant many exceptions $f_0,f_1,\ldots, f_c$ and $f_{d-c},\ldots,f_d$ from both ends. This covers (weighted) matchings, perfect matchings, and the subgraphs world transformed from the Ising model \cite{jerrum1993polynomial}, which is a Holant problem defined by constraint functions in the form $[1, \mu,1,\mu,\ldots]$ and $[1, 0, \lambda]$.
}

\subsection{Correlation decay}
Let $\Omega=(G(V,E),\{f_v\}_{v\in V})$ be a Holant instance where each $f_v:[q]^{\deg(v)}\rightarrow\mathbb{R}^+$ is a symmetric function with nonnegative real function values.
Let $\sigma\in[q]^E$ be a \concept{configuration} and $w(\sigma)=\prod_{v\in V}f_v\left(\sigma\mid_{E(v)}\right)$ be the weight of configuration $\sigma$.  

A configuration $\sigma\in[q]^E$ is \concept{feasible} if $w(\sigma)>0$. And for a configuration $\tau_\Lambda\in[q]^\Lambda$ on a subset $\Lambda\subseteq E$ of edges, we say that $\tau_\Lambda$ is feasible if there is a feasible $\sigma\in[q]^E$ agreeing with $\tau_\Lambda$ on $\Lambda$.

The \concept{Gibbs measure} is a probability distribution over all configurations, defined as $\mu(\sigma)=\frac{w(\sigma)}{\holant(\Omega)}$. To make the Gibbs measure well-defined, we require that each $f_v$ has nonnegative values and the Holant problem is feasible, i.e.~there exists a feasible configuration.

For a feasible $\sigma_\Lambda\in[q]^{\Lambda}$ on $\Lambda\subseteq E$, we use $\mu_e^{\sigma_\Lambda}$ to denote the marginal distribution at $e$ conditioning on the configuration of $\Lambda$ being fixed as $\sigma_\Lambda$.

\begin{definition}[Strong Spatial Mixing]\label{def-SSM}
A Holant problem $\Holant(\mathcal{G},\mathcal{F})$ 
has strong spatial mixing (SSM) if for any instance $\Omega=(G(V,E), \{f_v\}_{v\in V})$, any $e\in E$, $\Lambda\subseteq E$ and any two feasible configurations $\sigma_\Lambda,\tau_\Lambda\in[q]^\Lambda$, it holds that
\[
\left\|\mu^{\sigma_\Lambda}_e-\mu^{\tau_\Lambda}_e\right\|_{\mathrm{TV}}\le\mathrm{Poly}(|V|)\cdot\exp(-\Omega(\mathrm{dist}(e,\Delta))),
\]
where $\Delta\subseteq\Lambda$ is the subset on which $\sigma_\Lambda$ and $\tau_\Lambda$ differ, $\mathrm{dist}(e,\Delta)$ is the shortest distance from edge $e$ to any edges in $\Delta$, and $\|\cdot\|_{\mathrm{TV}}$ denotes the total variation distance.
\end{definition}

The strong spatial mixing defined for spin systems~\cite{weitz2006counting} is covered as special case.



\subsection{Tractable search}
In order to apply the self-reduction technique of Jerrum-Valiant-Vazirani \cite{jerrum1986random} for approximate counting, we also require that the following search problem is tractable:\\
\textbf{Input:} a Holant instance $\Omega=(G(V,E), \{f_v\}_{v\in V})$, and a configuration $\sigma_\Lambda\in[q]^{\Lambda}$ on $\Lambda\subseteq E$;\\
\textbf{Output:} a feasible $\tau\in[q]^E$ agreeing with $\sigma_\Lambda$ on $\Lambda$, or determines no such $\tau$ exists.


We call such property the \concept{tractable search} for $\Holant(\mathcal{G},\mathcal{F})$. 
We remark that this is a very natural assumption for approximate counting: for all known examples of approximate counting implied by mixing, the above search problem is easy or even trivial. The tractable search requirement of the general Holant framework is an analog to the specific $q\ge \Delta+1$ requirement for counting $q$-coloring.



\ifabs{}
{
The tractable search is related to the $\Holant^c$ framework.
For $i\in[q]$, let $\Delta_i$ denote the unary function which maps $i$ to 1 and all other $j\in[q]$ to 0. 
\begin{definition}
$\Holant^c(\mathcal{G},\mathcal{F})=\Holant(\mathcal{G},\mathcal{F}\cup\{\Delta_i\mid i\in[q]\})$.
\end{definition}
The $\Holant^c$ problem has significance in complexity of counting \cite{cai2011computational,caiholant}.
Assuming the tractable search for $\Holant(\mathcal{G},\mathcal{F})$ is equivalent to assuming the polynomial-time decision oracle (for existence of feasible configuration)  for $\Holant^c(\mathcal{G},\mathcal{F})$. 
}

\subsection{Local treewidth and planarity}
Our algorithm relies on the treewidth of graph and family of graphs with forbidden graph minors. We will not formally define these concepts excepting saying that the treewidth measures how similar a graph is to a tree. The formal definitions can be found in standard textbooks, e.g.~\cite{diestel2005graph}.

Graphs of \concept{bounded local treewidth} are precisely the
family of \concept{apex-minor-free} graphs, where  an apex graph has a vertex whose removal leaves a planar graph. In particular, $K_5$ and $K_{3,3}$ are apex graphs, therefore apex-minor-free graphs include planar graphs as a special case.

Let $\tw(G)$ denote the treewidth of graph $G$.
\begin{theorem}[\cite{demaine2004equivalence,eppstein1995subgraph}]\label{thm-local-tw}
Let $\mathcal{G}$ be an apex-minor-free family of graphs. For any $G(V,E)\in\mathcal{G}$ and $v\in V$, let $N_r(v)$ be the subgraph of $G$ induced by vertices whose distance to $v$ is at most $r$. Then $\mathrm{tw}(N_r(v))\le f(r)$ for some linear function $f$.
\end{theorem}
\ifabs{}
{
The following easy proposition states that representing spin systems as Holant problem on the incidence graphs does not violate the graph structure.
\begin{proposition}
Let $G$ be a graph and $\mathcal{I}_G$ be the bipartite incidence graph of $G$. Then
$\tw(G)=\tw(\mathcal{I}_G)$ and $\mathcal{I}_G$ is apex-minor-free if $G$ is apex-minor-free.
\end{proposition}
}

\subsection{Main results}

Our main results can be summarized by the following two theorems.

\begin{theorem}
There is an algorithm for $\Holant(\mathcal{G},\mathcal{F})$ with regular symmetric $\mathcal{F}$, 
whose running time is $2^{O(k)}\cdot\mathrm{poly}(n)$ for any $G\in\mathcal{G}$ of $n$ vertices and treewidth $k$.
\end{theorem}


\begin{theorem}
Assuming the tractable search, for Holant problem $\Holant(\mathcal{G},\mathcal{F})$ with apex-minor-free $\mathcal{G}$ and regular symmetric nonnegative $\mathcal{F}$, SSM implies FPTAS.
\end{theorem}

\section{Structure of regular functions}
In this section we introduce a new concept which characterizes the structure of outcomes of pinning as well as plays a key role in efficient computation of Holant problems.

Note that different $\sigma,\tau\in[q]^k$ (up to symmetry) may yield the same function after pinning $f$ with them. We classify members of $[q]^k$ into equivalence classes according to their effects of pinning on $f$ by introducing the following concept of \concept{peers}.

\begin{definition}[peers]
Let $f:[q]^d\rightarrow\mathbb{F}$ be a $d$-ary symmetric function. Let $0\le k\le d$ and $\tau\in[q]^k$. We define that $\Peer{\tau}{}{f}=g$ where $g:[q]^{k}\rightarrow\{0,1\}$ is a boolean symmetric function such that
\[
\forall\sigma\in[q]^k,\quad
g(\sigma)=
\begin{cases}
1 & \mbox{if }\Pin{\sigma}{}{f}=\Pin{\tau}{}{f},\\
0 & \mbox{otherwise}.
\end{cases}
\]
We may also interpret $\Peer{\tau}{}{f}$ as a set and write $\sigma\in\Peer{\tau}{}{f}$ if $\Peer{\tau}{}{f}(\sigma)=1$.
\end{definition}

The condition $\Pin{\sigma}{}{f}=\Pin{\tau}{}{f}$ defines an equivalence relation between $\sigma$ and $\tau$ (having the same effect of pinning on $f$). Then $\Peer{\tau}{}{f}$ is the indicator function of the equivalence class\ifabs{.}{ $\left\{\sigma\in[q]^k \mid \Pin{\sigma}{}{f}=\Pin{\tau}{}{f}\right\}$.} So we have the following easy but useful proposition.
\begin{proposition}\label{prop-peer}
$\sigma\in\Peer{\tau}{}{f}$ if and only if $\Peer{\sigma}{}{f}=\Peer{\tau}{}{f}$.
\end{proposition}

The peer images of an input uniquely determines the function value, specifically:
\begin{lemma}\label{lemma-peer-code}
Let $f:[q]^{d}\rightarrow\mathbb{F}$ be a symmetric function. Let $r\ge 1$, 
$d_1+d_2+\cdots+d_r=d$, and $\sigma_i,\tau_i\in[q]^{d_i}$ for $i=1,2,\ldots,r$. If $\Peer{\sigma_i}{}{f}=\Peer{\tau_i}{}{f}$ for all $i=1,2,\ldots,r$, then $f(\tau_1\tau_2\cdots\tau_r)=f(\sigma_1\sigma_2\cdots\sigma_r)$.
\end{lemma}
\ifabs{The proof of this lemma can be found in the full version of the paper in Appendix.}
{\begin{proof}
Due to Proposition \ref{prop-peer}, $\tau_i\in\Peer{\sigma_i}{}{f}$ for all $i=1,2,\ldots,r$.
We prove by induction on $r$. 
When $r=1$, for all $d\ge 0$ and $\sigma_1\in[q]^d$, $\Pin{\sigma_1}{}{f}$ is a trivial function $f(\sigma_1)$ (a function value) and $\Peer{\sigma_1}{}{f}$ is the equivalent class of all $\tau_1\in[q]^d$ which have the same $\Pin{\tau_1}{}{f}$ as $\Pin{\sigma_1}{}{f}$, i.e.~$f(\tau_1)=f(\sigma_1)$.

Assume the statement holds for all smaller $r$ and all $d$. 
Let $\sigma_i, \tau_i\in[q]^{d_i}, i=1,2,\ldots,r$ satisfy that $\tau_i\in\Peer{\sigma_i}{}{f}$ for all $i=1,2,\ldots,r$. Since $\tau_r\in\Peer{\sigma_r}{}{f}$, we have $\Pin{\tau_r}{}{f}=\Pin{\sigma_r}{}{f}$. Denote that $g=\Pin{\tau_r}{}{f}=\Pin{\sigma_r}{}{f}$. By definition of pinning, it holds that $f(\sigma_1\sigma_2\cdots\sigma_r)=g(\sigma_1\sigma_2\cdots\sigma_{r-1})$ and $f(\tau_1\tau_2\cdots\tau_r)=g(\tau_1\tau_2\cdots\tau_{r-1})$. Note that $g$ satisfies the induction hypothesis for $r-1$, which means that $g(\tau_1\tau_2\cdots\tau_{r-1})=g(\sigma_1\sigma_2\cdots\sigma_{r-1})$ as long as  $\tau_i\in\Peer{\sigma_i}{}{f}$ for all $i=1,2,\ldots,r-1$.
Therefore, for $\sigma_i, \tau_i\in[q]^{d_i}, i=1,2,\ldots,r$ that $\tau_i\in\Peer{\sigma_i}{}{f}$ for all $i=1,2,\ldots,r$, we have $f(\tau_1\tau_2\cdots\tau_r)=g(\tau_1\tau_2\cdots\tau_{r-1})=g(\sigma_1\sigma_2\cdots\sigma_{r-1})=f(\sigma_1\sigma_2\cdots\sigma_r)$. 
\end{proof}
}

Note that the outcome of $\Peer{\tau}{}{f}$ is still a symmetric function, so we can apply pinning and peering operations on it. We can define the peering closure which contains all possible outcomes of recursively applying peer operations on a function $f$.

\begin{definition}[peering closure]
Let $f:[q]^d\rightarrow\mathbb{F}$ be a $d$-ary symmetric function. Let $0\le k_r\le k_{r-1}\le\cdots\le k_1\le d$ and $\tau_i\in[q]^{k_i}, 1\le i\le r$.
We denote that 
\begin{align*}
\Peer{\tau_1,\tau_2,\ldots,\tau_r}{}{f}
&=
\Peer{\tau_r}{}{\Peer{\tau_1,\tau_2,\ldots,\tau_{r-1}}{}{f}}.
\end{align*}
The \concept{peering closure} of $f$, denoted by $\Peer{}{*}{f}$, is defined as
\[
\Peer{}{*}{f}
=
\left\{\Peer{\tau_1,\tau_2,\ldots,\tau_r}{}{f}
\mid 
\tau_i\in[q]^{k_i}\mbox{ for }1\le i\le r, 0\le k_r\le k_{r-1}\le\cdots\le k_1\le d, r\ge 1
\right\}.
\]
\end{definition}
Note that $\Peer{\tau}{}{f}$ is a boolean functions no matter what the range of $f$ is. A boolean function $g:[q]^d\rightarrow\{0,1\}$ can be seen equivalently as a set $\{\sigma\in[q]^d\mid g(\sigma)=1\}$. For two boolean functions $g$ and $h$ defined on the same domain $[q]^d$, we define the operations on boolean functions $g\cup h$, $g\cap h$, and $g\subseteq h$ according to the operations on their set representations. 

\ifabs{}{
Some useful properties of peering are better presented in this set language:
\begin{lemma}\label{lemma-peer-set}
Let $f:[q]^d\rightarrow\mathbb{F}$ be a symmetric function, and $g,h:[q]^d\rightarrow\{0,1\}$ be boolean symmetric functions. Let $\sigma\in[q]^\ell$ and $\tau\in[q]^k$ for arbitrary $0\le k\le \ell\le d$. We have
\begin{enumerate}
\item $\Peer{\tau}{}{f}\subseteq\Peer{\tau}{}{\Peer{\sigma}{}{f}}$;
\item $\left(\Peer{\tau}{}{g}\cap\Peer{\tau}{}{h}\right)\subseteq\Peer{\tau}{}{g\cup h}$.
\end{enumerate}
\end{lemma}
\begin{proof}
Both statements can be proved by directly expanding the definition of $\Peer{}{}{\cdot}$.
\begin{enumerate}
\item For any $\pi\in[q]^k$, suppose that $\pi\in\Peer{\tau}{}{f}$, which means $\Pin{\pi}{}{f}=\Pin{\tau}{}{f}$. Then for any $x\in[q]^{\ell-k}$, we have $\Pin{x\pi}{}{f}=\Pin{x\tau}{}{f}$, thus $\Pin{x\pi}{}{f}=\Pin{\sigma}{}{f}$ if and only if $\Pin{x\tau}{}{f}=\Pin{\sigma}{}{f}$, which is equivalent to that for any $x\in[q]^{\ell-k}$, $x\pi\in\Peer{\sigma}{}{f}$ if and only if $x\tau\in\Peer{\sigma}{}{f}$. This implies that $\Pin{\pi}{}{\Peer{\sigma}{}{f}}=\Pin{\tau}{}{\Peer{\sigma}{}{f}}$, which implies $\pi\in\Peer{\tau}{}{\Peer{\sigma}{}{f}}$. Therefore, $\Peer{\tau}{}{f}\subseteq\Peer{\tau}{}{\Peer{\sigma}{}{f}}$.
\item For any $\pi\in[q]^k$, suppose that $\pi\in\Peer{\tau}{}{g}\cap\Peer{\tau}{}{h}$, which implies that $\Pin{\tau}{}{g}=\Pin{\pi}{}{g}$ and $\Pin{\tau}{}{h}=\Pin{\pi}{}{h}$. Then for any $x\in[q]^{d-k}$, it holds that $x\tau\in g$ if and only if $x\pi\in g$ and $x\tau\in h$ if and only if $x\pi\in h$, thus $x\tau\in g\cup h$ if and only if $x\pi\in g\cup h$, which means $\Pin{\tau}{}{g\cup h}=\Pin{\pi}{}{g\cup h}$, thus $\pi\in\Peer{\tau}{}{g\cup h}$. Therefore, $\left(\Peer{\tau}{}{g}\cap\Peer{\tau}{}{h}\right)\subseteq\Peer{\tau}{}{g\cup h}$.
\end{enumerate}
\end{proof}
}

We then characterize all $k$-ary functions in peering closure by unions of boolean functions.
\begin{lemma}\label{lemma-peering-closure}
Let $f:[q]^d\rightarrow\mathbb{F}$ be a symmetric function.
For any $0\le k\le d$, every $k$-ary function $g\in\Peer{}{*}{f}$ can be represented as
\[
g=\bigcup_{i=1}^t\Peer{\sigma_i}{}{f}
\]
for some $\sigma_1,\ldots,\sigma_t\in[q]^k$.
\end{lemma}
\ifabs{}{
\begin{proof}
Recall that every $k$-ary function in $\Peer{}{*}{f}$ is in the following form:
\[
\Peer{\tau_1,\tau_2,\ldots,\tau_r}{}{f}=\Peer{\tau_r}{}{\Peer{\tau_{r-1}}{}{\cdots\left(\Peer{\tau_1}{}{f}\right)\cdots}}
\]
for some $r\ge 1$, $k=k_r\le k_{r-1}\le\cdots\le k_1\le d$, and $\tau_i\in[q]^{k_i}, 1\le i\le r$.

We then prove by induction on $r$ that $\Peer{\tau_1,\tau_2,\ldots,\tau_r}{}{f}=\Peer{\sigma_1}{}{f}\cup\cdots\cup\Peer{\sigma_t}{}{f}$ for some $\sigma_1,\ldots,\sigma_t\in[q]^k$ where $k$ is the arity of $\tau_r$. When $r=1$, this is trivially true. Assume the statement holds for all smaller $r$. Then $\Peer{\tau_1,\tau_2,\ldots,\tau_r}{}{f}=\Peer{\tau_r}{}{\Peer{\tau_1,\ldots,\tau_{r-1}}{}{f}}$. And due to the induction hypothesis, there exist $\sigma_1,\ldots,\sigma_t\in[q]^{k_{r-1}}$ such that
\[
\Peer{\tau_1,\ldots,\tau_{r-1}}{}{f}=\bigcup_{i=1}^t\Peer{\sigma_i}{}{f}.
\]
Due to Lemma \ref{lemma-peer-set}, we have
\begin{align*}
\bigcap_{i=1}^t\Peer{\tau_r}{}{\Peer{\sigma_i}{}{f}} 
&\subseteq
\Peer{\tau_r}{}{\bigcup_{i=1}^t\Peer{\sigma_i}{}{f}};\quad\mbox{and}\\
\Peer{\tau_r}{}{f} 
&\subseteq \Peer{\tau_r}{}{\Peer{\sigma}{}{f}}\mbox{ for any }\sigma\in[q]^{k_{r-1}}.
\end{align*}  
Combining these, we have  
\[
\Peer{\tau_r}{}{f}
\subseteq\Peer{\tau_r}{}{\bigcup_{i=1}^t\Peer{\sigma_i}{}{f}}
=
\Peer{\tau_r}{}{\Peer{\tau_1,\ldots,\tau_{r-1}}{}{f}}
=\Peer{\tau_1,\tau_2,\ldots,\tau_r}{}{f}.
\]
Note that this already implies the lemma, i.e.~there exist $t\ge 1$ and $\sigma_1,\ldots,\sigma_t\in[q]^{k}$ where $k$ is the arity of $\tau_r$ such that $\Peer{\tau_1,\tau_2,\ldots,\tau_r}{}{f}=\bigcup_{i=1}^t\Peer{\sigma_i}{}{f}$. To see this implication, by contradiction we assume that the statement is false. 
Then there must exist $\sigma_1,\sigma_2\in[q]^{k}$ such that $\sigma_1,\sigma_2\in\Peer{\sigma_1}{}{f}$ but $\sigma_1\in\Peer{\tau_1,\tau_2,\ldots,\tau_r}{}{f}$ and $\sigma_2\not\in\Peer{\tau_1,\tau_2,\ldots,\tau_r}{}{f}$.  
Recall that $\Peer{\sigma}{}{f}$ are equivalent classes.
Then $\sigma_1\in\Peer{\tau_1,\tau_2,\ldots,\tau_r}{}{f}=\Peer{\tau_r}{}{\Peer{\tau_1,\ldots,\tau_{r-1}}{}{f}}$ implies that $\Peer{\sigma_1}{}{\Peer{\tau_1,\ldots,\tau_{r-1}}{}{f}}=\Peer{\tau_r}{}{\Peer{\tau_1,\ldots,\tau_{r-1}}{}{f}}$. 
On the other hand,
we have shown that $\Peer{\sigma_1}{}{f}\subseteq \Peer{\sigma_1}{}{\Peer{\tau_1,\ldots,\tau_{r-1}}{}{f}}$.
However, due to the assumption we have $\sigma_2\not\in\Peer{\tau_1,\tau_2,\ldots,\tau_r}{}{f}=\Peer{\sigma_1}{}{\Peer{\tau_1,\ldots,\tau_{r-1}}{}{f}}$, which implies that $\Peer{\sigma_1}{}{f}\not\subseteq \Peer{\sigma_1}{}{\Peer{\tau_1,\ldots,\tau_{r-1}}{}{f}}$, a contradiction.

\end{proof}
}

\begin{lemma}\label{lemma-peer-regularity}
Let $f:[q]^d\rightarrow\mathbb{F}$ be a symmetric function.
If $f$ is \regular{C}, then 
\begin{enumerate}
\item for any $0\le k\le d$, the number of distinct $k$-ary functions in $\Peer{}{*}{f}$ is at most $2^C$;
\item for every $g\in\Peer{}{*}{f}$, $g$ is \regular{C}.
\end{enumerate} 
\end{lemma}
\ifabs{}{
\begin{proof}
First note that for any $0\le k\le d$, it holds that
\[
\left|\left\{ \Peer{\tau}{}{f} \mid \tau\in[q]^k\right\}\right|
=
\left|\left\{ \Pin{\tau}{}{f} \mid \tau\in[q]^k\right\}\right|.
\]
This is because they both count the number of equivalence classes defined by the equivalence relation that $\sigma$ and $\tau$ are equivalent if and only if $\Pin{\sigma}{}{f}=\Pin{\tau}{}{f}$. If $f$ is \regular{C}, then $\left|\left\{ \Pin{\tau}{}{f} \mid \tau\in[q]^k\right\}\right|\le C$, so we have $\left|\left\{ \Peer{\tau}{}{f} \mid \tau\in[q]^k\right\}\right|\le C$.

Due to Lemma \ref{lemma-peering-closure}, every $k$-ary function in $\Peer{}{*}{f}$ can be represented as a (boolean-functional) union $\bigcup_{i=1}^t\Peer{\sigma_i}{}{f}$ for some $\sigma_1,\ldots,\sigma_t\in[q]^k$. The number of different unions is obviously bounded by the size of power set of $\left\{ \Peer{\tau}{}{f} \mid \tau\in[q]^k\right\}$, which is at most $2^C$. Therefore the number of distinct $k$-ary functions in  $\Peer{}{*}{f}$ is at most $2^C$. The first part of the lemma is proved.


Let $0\le k\le d$ and $\tau\in[q]^{k}$.
By definition of $\Peer{}{}{\cdot}$, for any $0\le \ell\le k$, both $\left\{\Peer{\sigma}{}{f}\mid\sigma\in[q]^\ell\right\}$ and $\left\{\Peer{\sigma}{}{\Peer{\tau}{}{f}}\mid\sigma\in[q]^\ell\right\}$ are partitions of $[q]^\ell$. Due to Lemma \ref{lemma-peer-set}, for any $\sigma\in[q]^\ell$, we have 
\[
\Peer{\sigma}{}{f}\subseteq\Peer{\sigma}{}{ \Peer{\tau}{}{f}}.
\]
Thus $\left\{\Peer{\sigma}{}{f}\mid\sigma\in[q]^\ell\right\}$ is a refinement of $\left\{\Peer{\sigma}{}{\Peer{\tau}{}{f}}\mid\sigma\in[q]^\ell\right\}$. Therefore for any $0\le\ell\le k$, we have that
  \begin{align*}
    \left|\left\{\Pin{\sigma}{}{\Peer{\tau}{}{f}}\mid\sigma\in[q]^\ell\right\}\right|
    &=\left|\left\{\Peer{\sigma}{}{\Peer{\tau}{}{f}}\mid\sigma\in[q]^\ell\right\}\right|\\
    &\le\left|\left\{\Peer{\sigma}{}{f}\mid\sigma\in[q]^\ell\right\}\right|\\
    &=\left|\left\{\Pin{\sigma}{}{f}\mid \sigma\in[q]^\ell\right\}\right|,
  \end{align*}
which is bounded by $C$ if $f$ is \regular{C}.
Therefore, $\Peer{\tau}{}{f}$ is \regular{C} if $f$ is \regular{C}. The second part of the lemma is a trivial consequence of this fact.

\end{proof}
}

\ifabs{
Proofs of both lemma can be found in the full version of the paper in Appendix.}
{}

\section{The Separator Decomposition}
In this section we introduce a construction called separator decomposition which is essential to the efficient computation of Holant problems. We also show that the width of this decomposition is asymptotic equivalent to the treewidth of graphs and give an parameterized algorithm to construct separator decomposition.

Let $G(V,E)$ be an undirected graph and $H(U,F)$ be a subgraph of $G$. For any vertex set $R\subseteq V$, the \concept{vertex boundary of $R$ in $H$}, denoted by $\partial_H R$, is defined as $\partial_H R=\{u\in U\setminus R \mid \exists v\in R, uv\in F\}$. In particular, when $H=G$ we omit the subscript and denote the vertex boundary by $\partial R$.

\begin{definition}[Vertex Separator]
Let $G(V,E)$ be an undirected graph.
A vertex set $S\subseteq V$ is an \concept{$(X,Y)$-separator} in $G$, where $X,Y\subset V$ are two vertex sets, if 
\begin{enumerate}
\item $\{X,Y, S\}$ is a partition of $V$;
\item all $u\in X$ are disconnected from all $v\in Y$ in $G[V\setminus S]$.
\end{enumerate}
\end{definition}
Note that in the above definition we do not require $X$ and $Y$ to be connected or even nonempty. Any vertex set is an $(X,Y)$-separator if one of $X,Y$ is empty.

We now introduce the separator decomposition.

\begin{definition}[Separator Decomposition]
Let $G(V,E)$ be an undirected graph. A \concept{separator decomposition} of $G$ is a full binary tree $T$ (a rooted tree in which each node either has two children or is a leaf) with each node $i$ of $T$ associated with a pair $(V_i,S_i)$ such that $V_i,S_i\subseteq V$ and satisfy:
  \begin{enumerate}
  \item $V_r=V$ for the root $r$ in $T$ and $V_\ell=S_\ell=\emptyset$ for each leaf $\ell$ in $T$;
  \item for every non-leaf node $i$ and its two children $j,k$ in $T$,  $S_i$ is a $(V_j,V_k)$-separator in $G_i$, where $G_i=G[V_i]$ is the subgraph of $G$ induced by $V_i$. 
  \end{enumerate}
  The \concept{width} of a separator decomposition is the maximum of $|\partial V_i|$ and $|S_i|$ over all nodes $i$ in $T$.
\end{definition}
For a graph $G$ of $n$ vertices, a separator decomposition $T$ must have $O(n)$ nodes, because for each node $i\in T$ and its children $j,k$, we have $V_j$ and $V_k$ disjoint and $V_j,V_k\subset V_i$. 

We then connect the width of separator decomposition to the well-known treewidth of graphs, and give an algorithm for constructing separator decomposition with small width if there is such a decomposition for the input graph.

\begin{theorem}\label{thm-sw-tw}
Let $\mathrm{sw}(G)$ be the minimum width of all separator decompositions of $G$, and let $\tw(G)$ be the treewidth of $G$. Then 
\[
\frac{1}{8}\cdot\tw(G)+\frac{1}{24}\le \mathrm{sw}(G) \le 3\cdot\tw(G)+3.
\] 
And there is an algorithm which given as input a graph $G$ of $n$ vertices and treewidth $k$ constructs a separator decomposition of width at most $3(k+1)$ in time $2^{O(k)}\cdot\mathrm{poly}(n)$.
\end{theorem}
It is possible to have an $O(2^{f(k)}\cdot n)$ time algorithm for super-linear $f(k)$, which has better performance in the sense of parameterized complexity. However, our FPTAS on apex-minor-free graphs critically relies on that the time growth in treewidth $k$ is $2^{O(k)}$.

\ifabs{
The proof of the theorem can be found in the full version of the paper in Appendix.
}{
We will not give the formal definition of treewidth, but instead we will use the following notion of balanced separators to characterize treewidth .


\begin{definition}[Balanced separator]\label{def-balanced-cut}
Let $G(V,E)$ be an undirected graph and $W\subseteq V$. A vertex set $S\subseteq V$ is an \emph{balanced $W$-separator} in $G$ 
if $W\setminus S$ can be partitioned into $X$ and $Y$ such that $0<|X|,|Y|\le \frac{2}{3}|W|$, and all $u\in X$ are disconnected from all $v\in Y$ in $G[V\setminus S]$.
\end{definition}

A relation between treewidth and balanced separator is stated in next theorem implicit in \cite{robertson1995graph}.

\begin{theorem}[Robertson and Seymour \cite{robertson1995graph}]\label{thm-tw-separator}
Let $G=(V,E)$ be an undirected graph. If $\tw(G)\le k$, then for every $W\subseteq V$ of size at least $2k+3$
there exists a balanced $W$-separator of size at most $k+1$. Conversely, if for every $W\subseteq V$ of size $3k+1$ there exists a balanced $W$-separator of size at most $k+1$, then $\tw(G)\le 4k+1$.
\end{theorem}





We then use balanced separators to characterize the width of separator decompositions. 

\begin{lemma}\label{lemma-sw-bs}
Let $G=(V,E)$ be a graph of $n$ vertices. 
\begin{enumerate}
\item
If $G$ has a separator decomposition of width $s$, then for every $W\subseteq V$ of size at least $6s$ there is a balanced $W$-separator of size at most $ 2s$.
\item
If for every $W\subseteq V$ of size $6s$ there is a balanced $W$-separator of size at most $2s$, then $G$ contains a separator decomposition $T$ of width at most $6s$. And such $T$ can be constructed in time $2^{O(s)}\cdot\mathrm{poly}(n)$.
\end{enumerate}
\end{lemma}

\begin{proof}
We show the first part: existence of a separator decomposition of width $s$ implies that for every $W\subseteq V$ there is a balanced $W$-separator of size at most $ 2s$.

Let $T$ be a separator decomposition of $G$ of width $s$, with each node $i\in T$ associated with a vertex set $V_i$ and a separator $S_i$ of $G[V_i]$. It holds that $|\partial V_i|\le s$ and $|S_i|\le s$ for all $i\in T$.
  
Fix an arbitrary $W\subseteq V$ with $|W|\ge 6s$. Let $i$ be the node in $T$ of maximum depth satisfying $|V_i\cap W|>\frac{1}{2}|W|$. Such node $i$ always exists and must be a non-leaf since $V_r\cap W=W$ for the root $r$ and $V_{\ell}\cap W=\emptyset$ for every leaf $\ell$. Let $S=\partial V_i\cup S_i$. It holds that $|S|\le|\partial V_i|+|S_i|\le 2s$. We then show that $S$ is a balanced $W$-separator.

  
Let $j,k$ be the children of $i$ in $T$. Due to the maximality of the depth of node $i$, it holds that $|V_j\cap W|\le\frac{1}{2}|W|$ and $|V_k\cap W|\le\frac{1}{2}|W|$. We distinguish between the following two cases:

Case.1: $|V_j\cap W|<\frac{1}{3}|W|$ and $|V_k\cap W|<\frac{1}{3}|W|$. Let $X=(V_j\cap W)\cup (V_k\cap W)$ and $Y=W\setminus (X\cup S)$. We have $|X| \le |V_j\cap W|+|V_k\cap W| <\frac{2}{3}|W|$ and $|Y|=|W\setminus (X\cup S)|\le|W|-|V_i\cap W|\le\frac{1}{2}|W|<\frac{2}{3}|W|$. Moreover, both $X$ and $Y$ are nonempty, since
\begin{align*}
|X| 
&=|(V_j\cap W)\setminus S_i|
\ge |V_i\cap W|-|S_i|>\frac{1}{2}|W|-|S_i|\ge 3s-s>0,\\
|Y|
&=|W\setminus (X\cup S)|
\ge |W|-|X|-|S|>\frac{1}{3}|W|-|S|\ge2s-2s=0.
\end{align*}
It then remains to show that $X$ and $Y$ are separated by $S$ in $G$. It is easy to verify that $X\subset V_i$ and $Y\cap V_i=\emptyset$, thus every path $X$ to $Y$ must go through some vertex in $\partial V_i\subset S$. 
 
Case.2: $|V_j\cap W|\ge\frac{1}{3}|W|$ or $|V_k\cap W|\ge\frac{1}{3}|W|$.
  Without loss of generality, suppose that $|V_j\cap W|\ge|V_k\cap W|$ and $|V_j\cap W|\ge\frac{1}{3}|W|$. Let $X=V_j\cap W$, $Y=W\backslash (S\cup X)$.
  We have that $|X|=|V_j\cap W|\le\frac{1}{2}|W|$ and $|Y|\le|W|-|V_j\cap W|\le\frac{2}{3}|W|$. Also $X$ and $Y$ are nonempty since $|X|=|V_j\cap W|\ge\frac{1}{3}|W|>0$ and $|Y|\ge |W|-|X|-|S|>\frac{1}{2}|W|-|S|\ge3s-2s>0$.
  
  We then show that $X$ and $Y$ are separated by $S$ in $G$. It holds that $X\subset V_j$ and $Y\cap V_j=\emptyset$, thus every path $X$ to $Y$ must go through some vertex in $\partial V_j\subset (\partial V_i\cup S_i)=S$. 

Therefore, in both cases $S$ is a balanced $W$-separator. The first part of the lemma is proved.

\bigskip

We then prove the second part: if for every $W\subseteq V$ of size $6s$ there is a balanced $W$-separator of size at most $2s$, then $G$ contains a separator decomposition of width at most $6s$. We first prove the following claim.

\begin{claim*}
If for every $W\subseteq V$ of size $6s$ 
there is a balanced $W$-separator of size at most $2s$, then for any nonempty $R\subseteq V$ with $|\partial R|\le 6s$, there is a partition $\{X, Y, S\}$ of $R$ such that 
\begin{enumerate}
\item
$S$ is an $(X,Y)$-separator in $G[R]$ and $|S|\le 4s$;
\item
$|\partial X|, |\partial Y|\le 6s$.
\end{enumerate}
\end{claim*}
\begin{proof}
When $|R|\le 4s$, the claim holds by taking $S=R$ and $X=Y=\varnothing$. We then consider only the case $|R|>4s$.

Let $W\supseteq\partial R$ and $|W|=6s$.
Let $S'$ be a balanced $W$-separator in $G$ of size at most $2s$. Then $W\setminus S'$ can be partitioned into $X_W$ and $Y_W$ such that $X_W$ and $Y_W$ are disconnected in $G[V\setminus S']$, and $0<|X_W|,|Y_W|\le \frac{2}{3}|W\setminus S'|$. Since $G[V\setminus S']$ is disconnected,  we have that $V\setminus S'$ can be partitioned into  $X_V$ and $Y_V$ such that $X_W\subseteq X_V$, $Y_W\subseteq Y_V$, $X_V$ and $Y_V$ are disconnected in $G[V\setminus S']$, and $0<|X_V\cap W|,|Y_V\cap W|\le\frac{2}{3}|W\setminus S'|$.

We define that $S=S'\cap R$, $X=X_V\cap R$ and $Y=Y_V\cap R$. We then verify that they satisfy the requirement.

If $X=R$ (or $Y=R$) then $S=\emptyset$, and $X_V$ (respectively $Y_V$) contains more than $4s$ vertices in $W$, contradicting that $S'$ is a balanced $W$-separator.
And it follows from that $S'$ separates $X_V$ and $Y_V$ in $G$ that $S$ is an $(X,Y)$-separator in $G[R]$. It trivially holds that $|S|\le |S'|\le 2s$.

It is easy to see that $\partial X\subseteq X_W\cup S'$, therefore $|\partial X|\le |X_W|+|S'|\le 6s$. The same holds for $\partial Y$.
\end{proof}
Applying the above claim we can construct the separator decomposition $T$ for a graph $G(V,E)$ as follows. 
Initially, for the root $r$ of $T$, let $V_r=V$. For any current node $i\in T$, if $V_i\neq\emptyset$, set $R=V_i$ and apply the above claim to get an $(X,Y)$-separator of $S$ in $G[R]$ with desirable properties. Then create two children $i$ and $j$ in $T$, let $V_i=X$, $V_j=Y$, and recursively do the same thing for the two children.

It is easy to see that $T$ is a separator decomposition for $G$ of width at most $6s$. For $|W|=6s$, a balanced $W$-separator $S'$ can be found in time $2^{O(s)}\cdot\mathrm{poly}(n)$ by enumerating all $\{S,X,Y\}$ partitions of $W$ and running the standard network flow algorithm on $G[V\setminus S]$ to find a separator of $X$ and $Y$. This standard approach for finding balanced separator is also used in construction of tree decomposition  (see Chap.~11.2 of \cite{flum2006parameterized}). Therefore, $T$ can be constructed in time $2^{O(s)}\cdot\mathrm{poly}(n)$.
\end{proof}


Theorem \ref{thm-sw-tw} is proved by combining Lemma \ref{lemma-sw-bs} and Theorem \ref{thm-tw-separator}.
}



\section{Counting Algorithms}
This section contains three algorithms for Holant problem with regular symmetric constraint functions: a simple exponential-time dynamic programming algorithm; a fixed-parameter tractable (FPT) algorithm which uses the exponential-time algorithm as a subroutine; an FPTAS on apex-minor-free graphs via correlation decay which utilizes the FPT algorithm. 

With the construction of  separator decomposition, it is not hard to come up with a very natural $2^{O(\tw(G))}\cdot\mathrm{poly}(n)$-time dynamic programming algorithm for spin systems by enumerating the vertex boundaries and separators of components in the separator decomposition. However, the flexibility of Holant problems causes many new issues to the computation, which require more sophisticated algorithms to deal with.



\subsection{Simple $\exp(O(n))$-time algorithm}
Any Holant problem can be computed in time $\exp(O(|E|))$ by enumerating all configurations. For Holant problem with regular constraint functions, there is a simple dynamic programming algorithm which runs in time $\exp(O(|V|))$. This algorithm is used as a subroutine in our main algorithm.
\begin{theorem}\label{theorem-alg-simple-DP}
Let $\Omega=(G(V,E),\{f_v\}_{v\in V})$ be a Holant instance where $f_v:[q]^{\deg(v)}\rightarrow\mathbb{C}$ are symmetric functions. If all $f_v$ are \regular{C} for some constant $C>0$, then $\holant(\Omega)$ can be computed in time $(qC)^{O(|V|)}$.
\end{theorem}
\ifabs{
The algorithm is a simple dynamic programming, whose description and analysis can be found in the full version of the paper in Appendix.
}{
\newcommand{\func}[2]{\phi_{v_{{#1}}}^{({#2})}}

We enumerate the vertices in $V$ in an arbitrary order $v_1,v_2,\ldots,v_n$. Let $G_k(V_k,E_k)$ be a subgraph induced by the first $k$ vertices, i.e.~$V_k=\{v_i\mid 1\le i\le k\}$ and $E_k=\{v_iv_j\in E\mid 1\le i,j\le k\}$. 
For a $v\in V_k$, let $\deg_k(v)$ denote the degree of $v$ in $G_k$.
Fix any $1\le k\le n$. For $i=1,2,\ldots,k$, let $\func{i}{k}$ be symmetric functions at vertex $v_i$ in the form $\func{i}{k}:[q]^{\deg_{k}(v_i)}\rightarrow\mathbb{C}$. We define the following quantity:
\begin{align*}
Z\left(k,\left\{\func{i}{k}\right\}_{i=1,2,\ldots,k}\right)=\sum_{\sigma\in[q]^{E_k}}\prod_{i=1}^k\func{i}{k}\left(\sigma\mid_{E_k(v_i)}\right).
\end{align*}
In fact, each $Z\left(k,\left\{\func{i}{k}\right\}_{i=1,2,\ldots,k}\right)$ defines a new Holant problem on $G_k$. And the result of the original Holant problem is given by $\holant(\Omega)=Z(n, \{f_{v_i}\}_{i=1,2,\ldots n})$.  In general we have the following recursion:
\begin{align*}
Z\left(0,\emptyset\right)
&=1;\\
Z\left(k,\left\{\func{i}{k}\right\}_{i=1,2,\ldots,k}\right)
&=
\sum_{\sigma\in[q]^{E_{k}(v_k)}}\func{k}{k}(\sigma)\cdot Z\left(k-1,\left\{\func{i}{k-1}\right\}_{i=1,2,\ldots,k-1}\right),\\
\mbox{where }\,\,
\func{i}{k-1}
&=\begin{cases}
\Pin{\sigma(v_iv_k)}{}{\func{i}{k}} 
& \mbox{if }v_iv_k\in E_k,\\
\func{i}{k}
&\mbox{otherwise}.
\end{cases}
\end{align*}
This recursion separates the summation into different cases of configurations around $v_k$ and modifies the functions at the adjacent vertices according to the configuration. The correctness of the recursion can be easily verified by observing that the edge set $E_k$ is the disjoint union of $E_{k-1}$ and $E_k(v_k)$.

We then describe a dynamic programming algorithm which computes the Holant problem $Z(n, \{f_{v_i}\}_{i=1,2,\ldots n})$ in time $(qC)^{O(n)}$ if all $f_v$ are \regular{C}. The algorithm consists of two phases:
\begin{enumerate}
\item Preparation: For every $v\in V$, construct $\{\Pin{\sigma}{}{f_v}\mid \sigma\in[q]^\ell, 0\le \ell\le \deg(v)\}$ which contains all pinning outcomes of $f_v$. For symmetric $f_v$ this can be done in time polynomial of $\deg(v)$.
\item Dynamic programming: It is easy to see that for any $1\le k\le n$ and any $1\le i\le k$, function $\func{i}{k}$ is an outcome of a sequence of pinning of $f_{v_i}$, Moreover, it holds that
\[
\func{i}{k}\in\left\{\Pin{\sigma}{}{f_{v_i}}\mid \sigma\in[q]^{\deg_n(v_i)-\deg_k(v_i)}\right\},
\]
where the size of the set is bounded by $C$ since $f_{v_i}$ is \regular{C}. Therefore, $Z\left(k,\left\{\func{i}{k}\right\}_{i=1,2,\ldots,k}\right)$ for all $1\le k\le n$ can be stored in an $n\times C^{n}$ table, while each $\func{i}{k}$ can be retrieved from $\left\{\Pin{\sigma}{}{f_{v_i}}\mid \sigma\in[q]^{\deg_n(v_i)-\deg_k(v_i)}\right\}$ by an index ranging over $[C]$. It takes at most $q^{O(n)}$ time to fill each entry of the table. The total time complexity is $(qC)^{O(n)}$.
\end{enumerate}
}
\subsection{Fixed-parameter tractable algorithm}

\begin{theorem}\label{thm-alg-fpt}
Let $\Omega=(G(V,E),\{f_v\}_{v\in V})$ be a Holant instance where $f_v:[q]^{\deg(v)}\rightarrow\mathbb{C}$ are symmetric functions. If all $f_v$ are \regular{C} for some constant $C>0$, then $\holant(\Omega)$ can be computed in time $2^{O(\tw(G))}\cdot\mathrm{poly}(n)$ where $\tw(G)$ represents the treewidth of $G$. 
\end{theorem}

The $2^{O(\tw(G))}$ growth in treewidth is critical to our approximation algorithm on planar graphs introduced later, although any faster growth in treewidth is still fixed-parameter tractable.

\newcommand{\constr}[1]{\phi_{{#1}}}

\paragraph{The setup.}
Let $U\subseteq V$ be a set of vertices. Let $\partial U$ denote the vertex boundary of $U$, i.e.~$\partial U=\{v\in V\setminus U\mid \exists u\in U, uv\in E\}$. Let $H(U\cup\partial U, F)$ be the subgraph that $F=\{uv\in E\mid u,v\in U\mbox{ or }u\in U, v\in\partial U\}$, i.e.~$F$ includes all edges within $U$ and all edges crossing between $U$ and $\partial U$ (but not those edges with both endpoints in $\partial U$). 
For each $v\in\partial U$, let $\constr{v}:[q]^{\deg_H(v)}\rightarrow\{0,1\}$ be a boolean symmetric function. We define the following quantity:  
\begin{align}
Z\left(U,\left\{\constr{v}\right\}_{v\in\partial U}\right)
&=
\sum_{\sigma\in[q]^{F}}\prod_{v\in U\cup\partial U}g_v\left(\sigma\mid_{F(v)}\right),\label{eq:holant-fpt}\\
\mbox{where}\quad 
g_v
&=\begin{cases}
f_v & \mbox{if }v\in U,\\
\constr{v} & \mbox{if }v\in\partial U.
\end{cases}\nonumber
\end{align}
In fact $Z\left(U,\left\{\constr{v}\right\}_{v\in\partial U}\right)$ defines a Holant problem on graph $H(U\cup\partial U,F)$ with function $f_v$ at each $v\in U$ and boolean constraint $\constr{v}$ at each boundary vertex $v\in\partial U$. And the original Holant problem can be written as $\holant(\Omega)=Z(V,\emptyset)$.

\paragraph{The recursion.}
Suppose that $U$ can be partitioned into $S, U_1, U_2$ such that $S$ is a $(U_1,U_2)$-separator of $U$ in $G[U]$, where $G[U]$ is the subgraph of $G$ induced by $U$ (note that $S$ is not necessarily a separator of $H$). It is obvious that $\partial U_1\subseteq S\cup \partial U$ and $\partial U_2\subseteq S\cup \partial U$, where $\partial U_1$ and $\partial U_2$ are respectively the vertex boundaries of $U_1$ and $U_2$ in $G$.
For each $v\in S\cup\partial U$, let $d_0(v),d_1(v),d_2(v)$ denote the number neighbors of $v$ in $S\cup\partial U,U_1,U_2$ respectively, i.e.
\begin{align*}
d_0(v)
=
|\{u\in S\cup\partial U\mid uv\in F\}|,
\quad
d_1(v)
=
|\{u\in U_1\mid uv\in F\}|,
\quad
d_2(v)
=
|\{u\in U_2\mid uv\in F\}|.
\end{align*}
It holds that $d_0(v)+d_1(v)+d_2(v)=\deg_H(v)$.

\newcommand{\code}[2]{\phi_{{#1}}^{{#2}}}

For any $v\in S\cup\partial U$ and $i=0,1,2$, each tuple $\sigma\in[q]^{d_i(v)}$ can be mapped to a boolean function $\Peer{\sigma}{}{g_v}$ which indicates all tuples that have the same effect of pinning on $g_v$ as $\sigma$, where $g_v$ is still defined as that $g_v=f_v$ for $v\in S$ and $g_v=\constr{v}$ for $v\in\partial U$. We call $\Peer{\sigma}{}{g_v}$ the \concept{peer image} of $\sigma$ at $v$. 
For each $v\in S\cup\partial U$ and $i=0,1,2$, let $\mathcal{P}_v^i=\left\{\Peer{\sigma}{}{g_v}\mid \sigma\in[q]^{d_i(v)}\right\}$ be the range of peer images over all $\sigma\in[q]^{d_i(v)}$. Let $\code{}{}$ be a sequence indexed by $\code{v}{i}$ for $v\in S\cup\partial U$ and $i=0,1,2$, such that $\code{v}{i}\in\mathcal{P}_v^i$, i.e.~$\code{v}{i}$ is a peer image of some $\sigma\in[q]^{d_i(v)}$.

Due to Lemma \ref{lemma-peer-code}, for any $v\in S\cup\partial U$, for any $\sigma_i\in[q]^{d_i(v)}, i=0,1,2$, the value of $g_v(\sigma_0\sigma_1\sigma_2)$ is uniquely determined by the peer images $\Peer{\sigma_i}{}{g_v}, i=0,1,2$. 
For $\code{v}{i}\in\mathcal{P}_v^i, i=0,1,2$, we write $\tilde{g}_v(\code{v}{0},\code{v}{1},\code{v}{2})$ for the unique value of $g_v(\sigma_0\sigma_1\sigma_2)$ for all $(\sigma_0,\sigma_1,\sigma_2)$ with the same peer images $\Peer{\sigma_i}{}{g_v}=\code{v}{i}$ for $i=0,1,2$.

We then have the following recursion for $Z\left(U,\left\{\constr{v}\right\}_{v\in\partial U}\right)$ defined in \eqref{eq:holant-fpt}:
\begin{align}
Z\left(U,\left\{\constr{v}\right\}_{v\in\partial U}\right)
&=
\sum_{\substack{\code{}{}:\, \code{v}{i}\in \mathcal{P}_v^i\\\forall v\in S\cup\partial U\\ i=0,1,2}} 
Z_0(\code{}{})\cdot Z_1(\code{}{})\cdot Z_2(\code{}{}) \prod_{v\in S\cup\partial U}\tilde{g}_v\left(\code{v}{0},\code{v}{1},\code{v}{2}\right),\label{eq:holant-fpt-recursion}\\
\mbox{where}\quad Z_0(\code{}{})
&= \holant\left(H\left[S\cup\partial U\right],\left\{\code{v}{0}\right\}_{v\in S\cup\partial U}\right),\notag\\
Z_1(\code{}{})
&= Z\left(U_1,\left\{\code{v}{1}\right\}_{v\in\partial U_1}\right),\notag\\
Z_2(\code{}{})
&= Z\left(U_2,\left\{\code{v}{2}\right\}_{v\in\partial U_2}\right).\notag
\end{align}
Note that peer images $\code{v}{i}$ are boolean functions, thus $Z_0, Z_1$ and $Z_2$ are well defined. Also note that only $Z_1$ and $Z_2$ are recursions and $Z_0$ is a new well-defined Holant problem which can be directly computed by the simple algorithm of Theorem \ref{theorem-alg-simple-DP}. 

As an example, consider counting matchings, which is a Holant problem of regular constraint functions. The peer images $\code{v}{i}$ actually correspond to that vertex $v$ is matched or unmatched\footnote{We can ignore the overmatched cases in our discussion because they nullify the configuration.} by the corresponding subset of incident edges of $v$. The Holant problem $Z_0(\code{}{})= \holant\left(H\left[S\cup\partial U\right],\left\{\code{v}{0}\right\}_{v\in S\cup\partial U}\right)$ counts the number of perfect matchings of those vertices that claim to be matched in $H\left[S\cup\partial U\right]$.

\ifabs{The proof of correctness can be found in the full version of the paper in Appendix.}
{
We then prove that this recursion holds for $Z\left(U,\left\{\constr{v}\right\}_{v\in\partial U}\right)$ defined in \eqref{eq:holant-fpt}.
\begin{proof}
Since $S$ is a $(U_1,U_2)$-separator of $U$,
the edges in the original subgraph $H(U\cup\partial U, F)$ can be partitioned into five disjoint sets:
\begin{align*}
F_0:
&\quad\mbox{internal edges of }S\cup\partial U, \mbox{ i.e.~} F_0=\{uv\in F\mid u\in S, v\in S\cup\partial U\};\\
\mbox{for }i=1,2,\quad E_i:
&\quad\mbox{internal edges of }U_i, \mbox{ i.e.~} E_i=\{uv\in F\mid u,v\in U_i\};\\
\mbox{for }i=1,2,\quad F_i:
&\quad\mbox{boundary edges of }U_i, \mbox{ i.e.~} F_i=\{uv\in F\mid u\in U_i, v\in S\cup\partial U\}.
\end{align*}
Each vertex $v\in S\cup\partial U$ has precisely $d_i(v)$ adjacent edges in $F_i$ for $i=0,1,2$. And for $i=1,2$, $\partial U_i$ is precisely the set of vertices in $S\cup\partial U$ with positive $d_i(v)$. 
We can enumerate all configurations $\sigma\in[q]^{F_0\cup F_1\cup F_2}$ by enumerating legal local configurations $\sigma_v^i\in[q]^{F_i(v)}$ for each individual vertex $v\in S\cup\partial U$ and each $i=0, 1,2$, where being legal means that there exists a $\sigma\in[q]^{F_0\cup F_1\cup F_2}$ such that $\sigma\mid_{F_i(v)}=\sigma_v^i$ for all $v\in S\cup\partial U$ and $i=0, 1,2$.

\newcommand{\indi}[1]{\mathbf{1}_{{#1}}}

For a tuple $\sigma\in[q]^k$, we define the indicator function $\indi{\sigma}:[q]^k\rightarrow\{0,1\}$ as that $\indi{\sigma}(\tau)=1$ if and only if $\tau=\sigma$. Then we can rewrite \eqref{eq:holant-fpt} as follows:
\begin{align}
&\quad\,\, Z\left(U,\left\{\constr{v}\right\}_{v\in\partial U}\right)\notag\\
&=
\sum_{\sigma\in[q]^{F_0\cup F_1\cup F_2}} 
Z\left(U_1,\left\{\indi{\sigma\mid_{F_1(v)}}\right\}_{v\in\partial U_1}\right)
\cdot Z\left(U_2,\left\{\indi{\sigma\mid_{F_2(v)}}\right\}_{v\in\partial U_2}\right)
\prod_{v\in S\cup\partial U}g_v\left(\sigma\mid_{F_0(v)\cup F_1(v)\cup F_2(v)}\right)\notag\\
&=
\sum_{\substack{\text{legal }\sigma_v^i\in[q]^{F_i(v)}\\\forall v\in S\cup\partial U\\ i=0,1,2}}
Z\left(U_1,\left\{\indi{\sigma_v^1}\right\}_{v\in\partial U_1}\right)
\cdot Z\left(U_2,\left\{\indi{\sigma_v^2}\right\}_{v\in\partial U_2}\right)
\prod_{v\in S\cup\partial U}g_v\left(\sigma_v^0\sigma_v^1\sigma_v^2\right).\label{eq:recursion-proof-1}
\end{align}
In fact, $F_1$ can be partitioned into disjoint  $F_1(v)$ for $v\in\partial U_1$ and $F_2$ can be partitioned into disjoint $F_2(v)$ for $v\in\partial U_2$.
Thus for $i=1,2$ all local configurations $\{\sigma_v^i\in[q]^{d_i(v)}\}_{v\in S\cup\partial U}$ are legal.  A collection $\left\{\sigma_v^0\right\}_{v\in S\cup\partial U}$ of  local configurations of edges in $F_0$ of individual vertices is legal if and only if the Holant problem
$\holant\left(H[S\cup\partial U], \left\{\indi{\sigma_v^0}\right\}_{v\in S\cup\partial U}\right)$ 
has value 1 (it has only two possible values 0 or 1 as every indicator function has value 1 on exactly one input). Thus we have
\begin{align}
\eqref{eq:recursion-proof-1}
&=
\sum_{\substack{\sigma_v^i\in[q]^{F_i(v)}\\\forall v\in S\cup\partial U\\ i=0,1,2}}
\holant\left(H[S\cup\partial U], \left\{\indi{\sigma_v^0}\right\}_{v\in S\cup\partial U}\right)\notag\\
&\qquad\qquad\quad
\cdot Z\left(U_1,\left\{\indi{\sigma_v^1}\right\}_{v\in\partial U_1}\right)
\cdot Z\left(U_2,\left\{\indi{\sigma_v^2}\right\}_{v\in\partial U_2}\right)
\cdot \prod_{v\in S\cup\partial U}g_v\left(\sigma_v^0\sigma_v^1\sigma_v^2\right).\label{eq:recursion-proof-2}
\end{align}
For $v\in S\cup\partial U$ and $i=0,1,2$, fix $\code{v}{i}\in\mathcal{P}_v^i$, i.e.~$\code{v}{i}=\Peer{\sigma}{}{g_v}$ for some $\sigma\in[q]^{d_i(v)}$. 
We can group configurations $\{\sigma_v^i\}_{v\in S\cup\partial U, i=0,1,2}$ into equivalence classes $\{\sigma_v^i\in[q]^{d_i(v)}\mid \Peer{\sigma_v^i}{}{g_v}=\code{v}{i}\}$ according to their peer images. Due to Lemma \ref{lemma-peer-code}, configurations from the same class yields the same value of $g_v(\sigma_v^0\sigma_v^1\sigma_v^2)=\tilde{g}_v(\code{v}{0},\code{v}{1},\code{v}{2})$. Therefore,
\begin{align}
\eqref{eq:recursion-proof-2}
&=
\sum_{\substack{\code{}{}:\, \code{v}{i}\in \mathcal{P}_v^i\\\forall v\in S\cup\partial U\\ i=0,1,2}} 
\left(
\sum_{\substack{\sigma_v^0\in[q]^{F_0(v)}:\, \text{\textsc{Peer}}(\sigma_v^0,g_v)=\code{v}{0}\\ \forall v\in S\cup\partial U}}
\holant\left(H[S\cup\partial U], \left\{\indi{\sigma_v^0}\right\}_{v\in S\cup\partial U}\right)
\right)\notag\\
&\qquad\qquad\quad\cdot
\left(
\sum_{\substack{\sigma_v^1\in[q]^{F_1(v)}:\, \text{\textsc{Peer}}(\sigma_v^1,g_v)=\code{v}{1}\\ \forall v\in S\cup\partial U}}
Z\left(U_1,\left\{\indi{\sigma_v^1}\right\}_{v\in\partial U_1}\right)
\right)\notag\\
&\qquad\qquad\quad\cdot
\left(
\sum_{\substack{\sigma_v^2\in[q]^{F_2(v)}:\, \text{\textsc{Peer}}(\sigma_v^2,g_v)=\code{v}{2}\\ \forall v\in S\cup\partial U}}
Z\left(U_2,\left\{\indi{\sigma_v^2}\right\}_{v\in\partial U_2}\right)
\right)
\prod_{v\in S\cup\partial U}\tilde{g}_v(\code{v}{0},\code{v}{1},\code{v}{2}).\label{eq:recursion-proof-3}
\end{align}
And any peer image $\code{v}{i}\in\mathcal{P}_v^i$ is a boolean function which indicates all such $\sigma\in[q]^{d_i(v)}$ that have the same peer image $\Peer{\sigma}{}{g_v}=\code{v}{i}$, thus it is straightforward to verify the following identities:
\begin{align*}
\holant\left(H\left[S\cup\partial U\right],\left\{\code{v}{0}\right\}_{v\in S\cup\partial U}\right)
&=
\sum_{\substack{\sigma_v^0\in[q]^{F_0(v)}:\, \text{\textsc{Peer}}(\sigma_v^0,g_v)=\code{v}{0}\\ \forall v\in S\cup\partial U}}
\holant\left(H[S\cup\partial U], \left\{\indi{\sigma_v^0}\right\}_{v\in S\cup\partial U}\right),\\
Z\left(U_1,\left\{\code{v}{1}\right\}_{v\in\partial U_1}\right)
&=
\sum_{\substack{\sigma_v^1\in[q]^{F_1(v)}:\, \text{\textsc{Peer}}(\sigma_v^1,g_v)=\code{v}{1}\\ \forall v\in S\cup\partial U}}
Z\left(U_1,\left\{\indi{\sigma_v^1}\right\}_{v\in\partial U_1}\right),\\
Z\left(U_2,\left\{\code{v}{2}\right\}_{v\in\partial U_2}\right)
&=
\sum_{\substack{\sigma_v^2\in[q]^{F_2(v)}:\, \text{\textsc{Peer}}(\sigma_v^2,g_v)=\code{v}{2}\\ \forall v\in S\cup\partial U}}
Z\left(U_2,\left\{\indi{\sigma_v^2}\right\}_{v\in\partial U_2}\right).
\end{align*}
Substituting these identities back in \eqref{eq:recursion-proof-3}, we deduce the recursion \eqref{eq:holant-fpt-recursion}.
\end{proof}
}

\paragraph{The algorithm.}
We then describe an algorithm which computes $\holant(G(V,E), \{f_v\}_{v\in V})$ in time $2^{O(\tw(G))}\cdot\mathrm{poly}(|V|)$ if all $f_v$ are \regular{C} for some constant $C>0$. 
\begin{enumerate}
\item Constructing separator decomposition: By Theorem \ref{thm-sw-tw}, a separator decomposition $T$ of input graph $G$ of width $3(\tw(G)+1)$ can be constructed in time $2^{O(\tw(G))}\cdot\mathrm{poly}(n)$. 
\item Enumeration of peering closures: For every $v\in V$ and each $0\le k\le \deg(v)$, construct set $P(v,k)=\{\Peer{\sigma}{}{f_v}\mid \sigma\in[q]^k\}$ and all possible unions (defined on boolean functions) of members of $P(v,k)$. Due to Lemma \ref{lemma-peering-closure}, this will cover all functions in $\Peer{}{*}{f_v}$. And due to Lemma \ref{lemma-peer-regularity}, there are at most $2^C$ possible unions. The total time cost is polynomial because for symmetric functions all such $P(v,k)$ can be constructed in polynomial time.



\item Dynamic programming: 
Let $T$ be the separator decomposition constructed in the first step. Then each node $i\in T$ associated with a vertex set $V_i$ and a separator $S_i$ such that $|S_i|=O(\tw(G))$ and $|\partial V_i|=O(\tw(G))$, and if $j$ and $k$ are the two children of $i$ in $T$, $S_i$ is a $(V_j,V_k)$-separator in $G[V_i]$.
Apply the recursion \eqref{eq:holant-fpt-recursion} in this tree structure as follows: 
For each leaf $\ell\in T$, $V_\ell=\emptyset$, and $Z(V_{\ell},\emptyset)=1$;
and for each non-leaf node $i\in T$ with children $j$ and $k$ in $T$, $Z(V_i,\{\constr{v}\}_{v\in\partial V_i})$ is computed according to the recursion \eqref{eq:holant-fpt-recursion} by setting $U=V_i$, $U_1=V_j$ and $U_2=V_k$; in particular for the root $r$ of $T$, $V_r=V$ and $Z(V,\emptyset)=\holant(G(V,E), \{f_v\}_{v\in V})$. 

There are $O(|V|)$ nodes in a separator decomposition. For all $Z(V_i,\{\constr{v}\}_{v\in\partial V_i})$, every $\constr{v}$ is a boolean function in $\Peer{}{*}{f_v}$. Due to Lemma \ref{lemma-peer-regularity}, since $f_v$ is \regular{C}, once $V_i$ is fixed there are at most $2^C$ possible $\constr{v}$ for each $v\in\partial V_i$, where $|\partial V_i|=\tw(G)$. 
Therefore all  $Z(V_i,\{\constr{v}\}_{v\in\partial V_i})$ can be stored in a $O(n)\times 2^{O(C\cdot\tw(G))}$ table. 

Each entry of the dynamic programming table is filled according to the recursion \eqref{eq:holant-fpt-recursion}, which involves three nontrivial tasks:
\begin{enumerate}
\item (computing $Z_0$): Due to Lemma \ref{lemma-peer-regularity}, any $\code{v}{}\in\Peer{}{*}{f_v}$ is still \regular{C} since $f_v$ is \regular{C}, thus $Z_0=\holant(H_i[S_i\cup\partial V_i],\{\code{v}{}\}_{v\in S_i\cup\partial V_i})$ is a Holant problem with \regular{C} constraint functions which can be computed in time $(qC)^{|S_i\cup\partial V_i|}=(qC)^{O(\tw(G))}$ due to Theorem \ref{theorem-alg-simple-DP}.
\item (evaluating $\tilde{g}_v$): Each $\tilde{g}_v(\code{v}{0},\code{v}{1},\code{v}{2})$ can be easily evaluated by evaluating $g_v(\sigma_0\sigma_1\sigma_2)$ for arbitrary $\sigma_0\in \code{v}{0}, \sigma_1\in \code{v}{1}, \sigma_2\in \code{v}{2}$. 
\item (computing the sum): For every $v\in S_i\cup\partial V_i$, enumerate all $\le 2^C$ possible boolean functions of appropriate arity $\code{v}{}\in\Peer{}{*}{f_v}$. The total time is bounded by $2^{O(C\cdot\tw(G))}$ because $|S_i\cup\partial V_i|=O(\tw(G))$.
\end{enumerate}
The time cost for filling one entry of the dynamic programming table is bounded by $2^{O(\tw(G))}$ for constant $C$ and $q$.
\end{enumerate}
The total time cost for the above algorithm is bounded by $2^{O(\tw(G))}\cdot\mathrm{poly}(n)$ for constant $C$ and $q$.

\subsection{FPTAS from correlation decay}
\begin{theorem}\label{thm-FPTAS}
Assume the tractable search for the Holant problem  $\Holant(\mathcal{G},\mathcal{F})$ where $\mathcal{G}$ is an apex-minor-free graph family and $\mathcal{F}$ is a regular family of nonnegative symmetric functions.
The strong spatial mixing implies the existence of FPTAS for $\Holant(\mathcal{G},\mathcal{F})$.
\end{theorem}
\ifabs{
We briefly describe the main idea of the proof. By strong spatial mixing, an $O(\log n)$-ball around $v$ of arbitrarily fixed feasible boundary defines a marginal distribution at $v$ which approximates the true marginal distribution up to a $1/\mathrm{poly}(n)$ total variation. On an apex-minor-free graph, an $O(\log n)$-ball has treewidth $O(\log n)$ thus the exact marginal probability (with boundary fixed) can be computed in polynomial time. Then applying the self-reduction, we have an FPTAS for the Holant problem. The detailed proof is in the full version of the paper in Appendix.
}
{
Let $\Omega=(G(V,E),\{f_v\}_{v\in V})$ be a Holant instance, where $G$ is an apex-minor-free graph and all $f_v:[q]^{\deg(v)}\rightarrow\mathbb{R}^+$ are \regular{C} symmetric functions for some constant $C>0$. Let $\mu$ be the Gibbs measure defined by the Holant instance $\Omega$. 

Assume the tractable search and strong spatial mixing for $\Holant(\mathcal{G},\mathcal{F})$. We have the following lemma for approximation of marginal probabilities.

\begin{lemma}\label{lemma-alg-marginal}
Let $e\in E$. Let $\Lambda\subseteq E$ and $\tau_\Lambda\in[q]^\Lambda$ be a feasible configuration.  
The marginal probability $\mu_e^{\tau_\Lambda}(i)$ for any $i\in [q]$ can be approximated within any additive error $\epsilon$ in time $\mathrm{poly}(n,\frac{1}{\epsilon})$.
\end{lemma}
\begin{proof}
Let $N_r(e)=\{e'\in E\mid \mathrm{dist}(e,e')\le r\}$ be the \concept{$r$-neighborhood} of edge $e$ in $G$. Let $B_r(e)=\{uv\in E\setminus N_r(e)\mid \exists wv\in N_r(e)\}$ be the vertex boundary of the $r$-neighborhood.

Denote $\Delta= B_r(e)\setminus \Lambda$. When the strong spatial mixing holds, by Definition \ref{def-SSM}, for any $\sigma_\Delta,\pi_\Delta\in[q]^{\Delta}$ that both $(\tau_\Lambda,\sigma_\Delta)$ and $(\tau_\Lambda,\pi_\Delta)$ are feasible, it holds that 
$\left\|\mu_{e}^{\tau_\Lambda,\sigma_\Delta}-\mu_{e}^{\tau_\Lambda,\pi_\Delta}\right\|_{\mathrm{TV}}\le \mathrm{poly}(n)\cdot\exp(\Omega(-r))$.
Therefore for any $\sigma_\Delta\in[q]^{\Delta}$ that $(\tau_\Lambda,\sigma_\Delta)$ is feasible, we have
\begin{align}
\left\|\mu_{e}^{\tau_\Lambda}-\mu_{e}^{\tau_\Lambda,\sigma_\Delta}\right\|_{\mathrm{TV}}\le \mathrm{poly}(n)\cdot\exp(\Omega(-r)), \label{eq:fptas-ssm}
\end{align} 
because $\mu_e^{\tau_\Lambda}$ is a linear combination of all such $\mu_{e}^{\tau_\Lambda,\sigma_\Delta}$.

Note that the joint configuration $(\tau_\Lambda,\sigma_\Delta)$ fixes the boundary $B_r(e)$. Thus for each $i\in[q]$, the marginal probability $\mu_{e}^{\tau_\Lambda,\sigma_\Delta}(i)$ can be computed precisely from the $r$-neighborhood as follows:

Let $W$ be the set of incident vertices of $N_r(e)$ and $F=N_r(e)\setminus \Lambda$. Let $H(W,F)$ be the subgraph formed by removing edges fixed by $\sigma_\Lambda$ from the $r$-neighborhood. For $i\in[q]$, let $e\mapsto i$ denote the configuration on $\{e\}$ that simply assigns value $i$ to edge $e$. We have that
\begin{align}
\mu_{e}^{\tau_\Lambda,\sigma_\Delta}(i)
&=
\frac{\holant\left(H'(W,F\setminus\{e\}),\left\{f_v^{\tau_\Lambda,\sigma_\Delta,e\mapsto i}\right\}_{v\in W}\right)}{\holant\left(H(W,F),\left\{f_v^{\tau_\Lambda,\sigma_\Delta}\right\}_{v\in W}\right)},\label{eq:marginal-holant}\\
& \mbox{where }\quad
f^{\tau}_v=\Pin{\tau_v}{}{f_v}\mbox{ and }\tau_v=\tau\mid_{F(v)}. \notag
\end{align}
The correctness of the equation and the well-defined-ness of the new Holant problems are easy to verify.


Since the original graph $G$ is apex-minor-free, due to Theorem \ref{thm-local-tw} we have $\tw(H)=O(r)$. Since all original $f_v$ are \regular{C}, then trivially all $f_v^\tau$ are \regular{C} since they are just results of pinning $f_v$. Then applying Theorem \ref{thm-alg-fpt}, the new Holant problems defined in \eqref{eq:marginal-holant} can be computed in time $2^{O(r)}\cdot\mathrm{poly}(n)$. Therefore, the marginal probability with boundary condition $\mu_{e}^{\tau_\Lambda,\sigma_\Delta}(i)$ for any $i\in[q]$ can be computed precisely in time $2^{O(r)}\cdot\mathrm{poly}(n)$ once a feasible $(\tau_\Lambda,\sigma_\Delta)$ is given. 

Due to the tractable search for $\Holant(\mathcal{G},\mathcal{F})$, given any feasible $\tau_\Lambda\in[q]^\Lambda$ it is possible to efficiently choose an arbitrary feasible $\sigma\in[q]^E$ agreeing with $\tau_\Lambda$. Thus a feasible $\sigma_\Delta\in[q]^\Delta$ can be efficiently constructed by restricting the aforementioned $\sigma$ on $\Delta$.
Due to \eqref{eq:fptas-ssm}, the original marginal probability $\mu_{e}^{\tau_\Lambda}(i)$ for any $i\in[q]$ can be approximated within an additive error $\epsilon$ in time $\mathrm{poly}(n,\frac{1}{\epsilon})$ by choosing appropriate $r=O(\log n+\log\frac{1}{\epsilon})$.
\end{proof}


With the above lemma, we can apply the standard self-reduction procedure to obtain the FPTAS for $\Holant(\mathcal{G},\mathcal{F})$.

Let $\tau\in[q]^E$ be a feasible configuration, i.e.~the Gibbs measure $\mu(\tau)>0$. Enumerate edges in $E$ as $e_1,e_2,\ldots,e_m$. For each $0\le k\le m$, let $E_k=\{e_1,\ldots,e_k\}$, $\tau_k\in[q]^{E_k}$ be consistent with $\tau$ on $E_k$, and $p_k=\mu^{\tau_{k-1}}_{e_k}(\tau(e_k))$.
The following identity hold for $\mu(\tau)$:
\begin{align*}
\mu(\tau)
&=
\prod_{k=1}^m \Pr_{\sigma\in[q]^E}\left[\sigma(e_k)=\tau(e_k)\mid \sigma(e_i)=\tau(e_i), 1\le i\le k-1\right]
=\prod_{k=1}^{m}p_k.
\end{align*}
On the other hand,  $\mu(\tau)=\frac{\prod_{v\in V}f_v\left(\tau\mid_{E(v)}\right)}{\holant(\Omega)}$.
Thus $\holant(\Omega)=\frac{\prod_{v\in V}f_v\left(\tau\mid_{E(v)}\right)}{\prod_{k=1}^{m}p_k}$. If for each $k$: (1) $p_k$ can be approximated in an additive error $\epsilon$; and (2) $p_k>0$ is a constant, then the product $\prod_{k=1}^{m}p_k$ can be approximated within a multiplicative factor $(1\pm O(n\epsilon))$. While (1) is guaranteed by  Lemma \ref{lemma-alg-marginal}, (2) can be achieved by trying $\mu^{\tau_{k-1}}_{e_k}(i)$ for all $i\in[q]$ and choosing $\tau(e_k)$ to be the $i$ with the largest returned value. This gives us an FPTAS for the Holant problem.

We then can directly apply any known strong spatial mixing result to get the FPTAS.
For example, combining with the result of~\cite{goldberg2005strong}, we have the following corollary.
\begin{corollary}
There exists an FPTAS for counting $q$-coloring on apex-minor-free triangle-free graphs of maximum degree at most $\Delta$ if $q>\alpha\Delta-\gamma$ where $\alpha\approx1.76322$ is the solution to $\alpha^\alpha=e$ and $\gamma=\frac{4\alpha^3-6\alpha^2-3\alpha+4}{2(\alpha^2-)}\approx0.47031$.
\end{corollary}
Note that although the original result of~\cite{goldberg2005strong} is proved for single-site strong spatial mixing where the boundaries differ on only one vertex, it implies our definition of strong spatial mixing on any finite graphs.
}

\newcommand{\Disc}[2]{\mathrm{Disc}_{{#1}}\left({#2}\right)}

\newcommand{\vbond}{\partial}
\newcommand{\ebond}{\delta}
\newcommand{\bmu}[2]{\mu_{{#1}}^{{#2}}}
\newcommand{\mbmu}[3]{\mu_{{#1},{#3}}^{{#2}}}
\newcommand{\recolor}[3]{{#1}_{#2}^{{#3}}}

\section{Correlation Decay}
\ifabs{
In this section we apply the recursive coupling technique~\cite{goldberg2005strong} to Holant problems, and prove strong spatial mixing for subgraphs world~\cite{jerrum1993polynomial} and ferromagnetic Potts model. The following algorithmic results are implied by Theorem \ref{thm-FPTAS}.
}{
In this section we apply the recursive coupling technique~\cite{goldberg2005strong} to Holant problems, and prove strong spatial mixing for subgraphs world~\cite{jerrum1993polynomial} and ferromagnetic Potts model. The algorithmic implications of these correlation decay results are presented in the end of this section.

\subsection{Recursive coupling on Holant Problems}
Consider a Holant problem $\Holant(\mathcal{G},\mathcal{F})$ and an instance $\Omega=(G(V,E),\{f_v\}_{v\in V})$. Let $R\subseteq E$, called \concept{region}. Define $\ebond R=\{uv\in E\mid uv\not\in R,\exists uw\in R\}$ the edge boundary of $R$. Define $V_R=\{v\in V\mid \exists uv\in R\}$. 

A \concept{boundary configuration} of $R$, is a $\sigma\in[q]^{\ebond R}$. For every boundary configuration $\sigma\in[q]^{\ebond R}$ and a configuration $\eta\in[q]^R$ of the region $R$, define the \concept{regional weight} as
\[
w_R^\sigma(\eta)=\prod_{v\in V_R}f_v(\eta\mid_{R(v)}\sigma\mid_{(\ebond R)(v)}),
\]
where $\eta\mid_{R(v)}$ is the restriction of $\eta$ on the edges in $R$ incident to $v$, $\sigma\mid_{(\ebond R)(v)}$ is the restriction of $\sigma$ on edges in $\ebond R$ incident to $v$, and  $f_v(\sigma\mid_{R(v)}\eta\mid_{(\ebond R)(v)})$ evaluates $f_v$ on the concatenation of them.

We say that a boundary configuration $\sigma\in[q]^{\ebond R}$ is \concept{$R$-feasible} if there exists an $\eta\in[q]^R$ such that $w_R^\sigma(\eta)>0$. For $R$-feasible boundary configuration $\sigma\in[q]^{\ebond R}$, a \concept{regional Gibbs measure} $\bmu{R}{\sigma}$ over $[q]^R$ can be defined as that $\bmu{R}{\sigma}(\eta)=\frac{w_R^\sigma(\eta)}{\sum_{\pi\in[q]^R}w_R^\sigma(\pi)}$ for each $\eta\in[q]^R$. For $R'\subseteq R$, let $\mbmu{R}{\sigma}{R'}$ denote the marginal distribution of $\bmu{R}{\sigma}$ on $R'$, and we write that $\mbmu{R}{\sigma}{e}=\mbmu{R}{\sigma}{\{e\}}$.

\begin{definition}
Let $R\subseteq E$, $R'\subseteq R$, and $\sigma,\tau\in[q]^{\ebond R}$ be two $R$-feasible boundary configurations. Let $\Psi(R,\sigma,\tau)$ be a coupling of $\bmu{R}{\sigma},\bmu{R}{\tau}$. Define the discrepancy of $\Psi(R,\sigma,\tau)$ on region $R'\subseteq R$ as
  \[
  \Disc{\Psi(R,\sigma,\tau)}{R'}=\Pr_{(\eta,\eta')\sim\Psi(R,\sigma,\tau)}[\eta\mid_{R'}\ne\eta'\mid_{R'}]
  \]
If $R'=\{e\}$, we write that $\Disc{\Psi(R,\sigma,\tau)}{e}=\Disc{\Psi(R,\sigma,\tau)}{\{e\}}$.
\end{definition}

\begin{definition}
For any two $R$-feasible boundary configurations $\sigma,\tau\in[q]^{\ebond R}$ differing on $\Delta\subseteq\ebond R$, a sequence of $R$-feasible boundary configurations $\sigma_1,\sigma_2,\dots,\sigma_t$ is called a \concept{feasible path} from $\sigma$ to $\tau$ if $\sigma=\sigma_1, \tau=\sigma_t$ and $\sigma_i,\sigma_{i+1}$ differ only at one edge $e\in\Delta$ for each $1\le i<t$. Let $T(\sigma,\tau)$ be the minimum such $t$, or be $\infty$ if no such path exists.
\end{definition}

\begin{lemma}\label{lemma-holant-ssssm}
Let $\Lambda\subset E$ and $\sigma,\tau\in[q]^{\Lambda}$ be two feasible configurations differing on $\Delta\subseteq\Lambda$. Let $R=E\setminus\Lambda$ and $e\in R$. 
There exist two $R$-feasible boundary configurations $\sigma',\tau'\in[q]^{\ebond R}$ differing only on edges in $\Delta$ such that
\begin{align*}
\left\|\mu_e^\sigma-\mu_e^\tau\right\|_{\mathrm{TV}}
&=
\left\|\mbmu{R}{\sigma'}{e}-\mbmu{R}{\tau'}{e}\right\|_{\mathrm{TV}}
\le
T(\sigma',\tau')
\cdot
\max_{\substack{\sigma_1,\sigma_2\in[q]^{\ebond R}\\\text{ differ on an }e'\in\Delta}}\|\mbmu{R}{\sigma_1}{e}-\mbmu{R}{\sigma_2}{e}\|_{\mathrm{TV}}\\
&\le
T(\sigma',\tau')
\cdot
\max_{\substack{\sigma_1,\sigma_2\in[q]^{\ebond R}\\\text{ differ on an }e'\in\Delta}}\Disc{\Psi(R,\sigma_1,\sigma_2)}{e},
\end{align*}
for arbitrary coupling $\Psi(R,\sigma_1,\sigma_2)$ of $\bmu{R}{\sigma_1},\bmu{R}{\sigma_2}$.
\end{lemma}
\begin{proof}
Let $\sigma',\tau'\in[q]^{\ebond R}$ consistent with $\sigma,\tau$ on $\ebond R$ respectively. It is easy to check that $\sigma',\tau'$ satisfy the equation. 
Let $\sigma'=\sigma_1,\sigma_2,\ldots,\sigma_t=\tau'$ be the feasible path from $\sigma'$ to $\tau'$ of length $t=T(\sigma',\tau')$ and $\left\|\mbmu{R}{\sigma'}{e}-\mbmu{R}{\tau'}{e}\right\|_{\mathrm{TV}}$ can be bounded by applying path coupling to $\mbmu{R}{\sigma_i}{e},\mbmu{R}{\sigma_{i+1}}{e}$ for $1\le i<T(\sigma',\tau')$. The last inequality is due to the coupling lemma.
\end{proof}

The above lemma reduce the strong spatial mixing to the discrepancy witnessed by a coupling of $\bmu{R}{\sigma},\bmu{R}{\tau}$ with $\sigma,\tau$ disagreeing at one edge. We then show a way to recursively construct the coupling.  This method is proposed by Goldberg \emph{et al.} in~\cite{goldberg2005strong} on colorings.

\paragraph{The recursive coupling.}
Let $R\subseteq E$ and $e_0\in R$. Let $\sigma,\tau\in[q]^{\ebond R}$ be any two $R$-feasible boundary configurations that differ at only one edge $e\in\ebond R$. 
Let $R(e)$ be set of edges in $R$ incident to $e$. Let $\Psi_{R(e)}(R,\sigma,\tau)$ be a coupling of marginal distributions $\mbmu{R}{\sigma}{R(e)},\mbmu{R}{\tau}{R(e)}$. 
A coupling $\Psi(R,\sigma,\tau)$ of regional Gibbs measures $\bmu{R}{\sigma},\bmu{R}{\tau}$ can be recursively constructed by the local coupling rule $\Psi_{R(e)}(R,\sigma,\tau)$: 
Let $(\eta,\eta')\in[q]^R\times[q]^R$ denote the pair sampled from $\Psi(R,\sigma,\tau)$.

\begin{enumerate}
\item (Base case) If $e_0\in R(e)$, sample $(\eta\mid_{R(e)},\eta'\mid_{R(e)})$ according to $\Psi_{R(e)}(R,\sigma,\tau)$ and arbitrarily sample the rest of $(\eta,\eta')$ conditioning on $(\eta\mid_{R(e)},\eta'\mid_{R(e)})$ as long as $(\eta,\eta')$ is a faithful coupling of $\bmu{R}{\sigma},\bmu{R}{\tau}$. 
If $|R(e)|=0$, in which case $e$ and $e_0$ are disconnected in $G[V_R]$, sample $(\eta,\eta')$ such that $(\eta(e_0),\eta'(e_0))$ is perfectly coupled.
\item (General case) $|R(e)|>0$ and $e_0\not\in R(e)$. Sample $(\eta\mid_{R(e)},\eta'\mid_{R(e)})$ according to $\Psi_{R(e)}(R,\sigma,\tau)$. Let $x,y\in[q]^{R(e)}$ be two configurations that $x=\eta\mid_{R(e)},y=\eta'\mid_{R(e)}$. Construct new region and boundaries as: $R'=R\setminus R(e)$; $\recolor{\sigma}{}{x}\in[q]^{\ebond R'}$ agrees with $\sigma$ on common edges and $\recolor{\sigma}{}{x}\mid_{R(e)}=x$; and $\recolor{\tau}{}{y}\in[q]^{\ebond R'}$ agrees with $\tau$ on common edges and $\recolor{\tau}{}{y}\mid_{R(e)}=y$. The rest of $(\eta,\eta')$ is sampled from a coupling $\Psi(x,y)$ of $\bmu{R'}{\recolor{\sigma}{}{x}}, \bmu{R'}{\recolor{\tau}{}{y}}$. If $x=y$, then $\recolor{\sigma}{}{x}=\recolor{\tau}{}{y}$ and $\Psi(x,y)$ is a perfect coupling. If $x\ne y$, let $\recolor{\sigma}{1}{(x,y)},\dots,\recolor{\sigma}{t}{(x,y)}$ be a feasible path from $\recolor{\sigma}{}{x}$ to $\recolor{\tau}{}{y}$ of length $t=T(\recolor{\sigma}{}{x},\recolor{\tau}{}{y})$. Let $\Psi(x,y)$ be the composition of coupling $\Psi\left(R',\recolor{\sigma}{i}{(x,y)},\recolor{\sigma}{i+1}{(x,y)}\right)$, $i=1,2,\dots,t-1$, in the same manner as path coupling, where each $\Psi\left(R',\recolor{\sigma}{i}{(x,y)},\recolor{\sigma}{i+1}{(x,y)}\right)$ can be recursively defined as $\recolor{\sigma}{i}{(x,y)}$ and $\recolor{\sigma}{i+1}{(x,y)}$ differ at only one edge. 
It is easy to verify that $\Psi(x,y)$ is a coupling of $\bmu{R}{\recolor{\sigma}{}{x}},\bmu{R}{\recolor{\tau}{}{y}}$. This complete the construction of $\Psi(R,\sigma,\tau)$.
\end{enumerate}

The following lemma is similar to the one proved in \cite{goldberg2005strong} for the recursive coupling constructed on spin systems.

\begin{lemma}\label{lem:holantcouple}
For the coupling $\Psi(R,\sigma,\tau)$ constructed as above, we have
\begin{align*}
&\quad\, \Disc{\Psi(R,\sigma,\tau)}{e_0}\\
&\le\sum_{\substack{x,y\in[q]^{R(e)}\\x\ne y}}\Pr_{\substack{(\eta,\eta')\sim\\\Psi_{R(e)}(R,\sigma,\tau)}}[\eta\mid_{R(e)}=x\land\eta'\mid_{R(e)}=y]
\cdot\sum_{i=1}^{T(\recolor{\sigma}{}{x},\recolor{\tau}{}{y})-1}\Disc{\Psi\left(R\setminus R(e),\recolor{\sigma}{i}{(x,y)},\recolor{\sigma}{i+1}{(x,y)}\right)}{e_0}\\
    &\le\Disc{\Psi_{R(e)}(R,\sigma,\tau)}{e}\cdot\max_{\substack{\sigma_1,\sigma_2\in[q]^{\ebond (R\setminus R(e))}\\\text{differ on }R(e)}}T(\sigma_1,\sigma_2)\cdot\max_{\substack{\sigma_1,\sigma_2\in[q]^{\ebond (R\setminus R(e))}\\\text{differ on an }e\in R(e)}}\Disc{\Psi(R\setminus R(e),\sigma_1,\sigma_2)}{e_0}
  \end{align*}
\end{lemma}
\begin{proof}
The lemma follows directly from our construction of $\Psi(R,\sigma,\tau)$.
\end{proof}

A standard choice of $\Psi_{R(e)}(R,\sigma,\tau)$ is the one that $\Pr_{(\eta,\eta')\sim\Psi_{R(e)}(R,\sigma,\tau)}\left[\eta\mid_{R(e)}=\eta'\mid_{R(e)}\right]$ is maximized, i.e.~$\Disc{\Phi_{R(e)}(R,\sigma,\tau)}{R(e)}$ is minimized. 

\subsection{The subgraphs world problem}
The subgraphs world model used in~\cite{jerrum1993polynomial} for developing FPRAS for the ferromagnetic Ising model, is a counting problem computationally equivalent to the Ising model under holographic transformation.
\begin{definition}[Subgraphs world \cite{jerrum1993polynomial}]
The \concept{subgraphs world} with parameters $(\lambda,\mu)$ defined as follows. Let $G=(V,E)$ be an undirected graph. The subgraphs world partition function is defined as:
\[
  Z_{\mathrm{sub}}(G)=\sum_{X\subseteq E}\mu^{|\mathrm{odd}(X)|}\lambda^{|X|},
\]
where $\mathrm{odd}(X)$ denotes the set of vertices with odd degree in the subgraph $(V,X)$.
\end{definition}

The subgraphs world with parameter $(\lambda,\mu)$ can be interpreted as a Holant problem on incident graph $\mathcal{I}_G$ as follows. 
The incident graph $\mathcal{I}_G$ has left vertex set $V$ and right vertex set $E$, and for each $v\in V$ and $e\in E$, $(v,e)$ is an edge in $\mathcal{I}_G$ if $e$ is incident to $v$ in $G$. The function on each left vertex $v$ is $[1, \mu,1,\mu,\ldots]$ and the function on each right vertex $e$ is $[1,0,\lambda]$. 
Let $\Omega$ be the Holant instance defined as above.
It is easy to verify that all functions in $\Omega$ are \regular{3} and $Z_{\mathrm{sub}}(G)=\holant(\Omega)$. 

For convenience of analysis, we consider the following equivalent Holant problem which is defined on the original graph $G$ instead of the incidence graph.
Let $\Omega'=(G(V,E),\{f_v\}_{v\in V})$ be a Holant instance where each $f_v=[f_0,f_1,\dots,f_{\deg(v)}]$ has that $f_k=\mu\lambda^{k/2}$ if $k$ is odd and $f_k=\lambda^{k/2}$ if $k$ is even. It is easy to see that $\holant(\Omega)=\holant(\Omega')$ and also the two Holant problems have exact the same Gibbs measure. Thus it is sufficient to analyze the correlation decay on this new Holant problem.


\begin{theorem}\label{thm-subgraph-world-decay}
Let $\Holant(\mathcal{G},\mathcal{F})$ be defined by the subgraphs world of parameter $(\mu,\lambda)$ with $0<\mu,\lambda<1$ on graphs with degree bound $\Delta$.
If $\Delta<\frac{(1+\lambda\mu^2)^2}{1-\mu^2}$, then $\Holant(\mathcal{G},\mathcal{F})$ has strong spatial mixing.
\end{theorem}
 \begin{proof}
Let $R\subseteq E$ be a region and $\sigma,\tau\in\{0,1\}^{\ebond R}$ be two $R$-feasible boundary configurations differing at $e\in\ebond R$ satisfying $\sigma(e)=0$ and $\tau(e)=1$. The regional weights $w^\sigma_R,w^\tau_R$ and regional Gibbs measures $\bmu{R}{\sigma},\bmu{R}{\tau}$ can be defined accordingly.

Let $\Psi_{R(e)}(R,\sigma,\tau)$ be the coupling of $\mbmu{R}{\sigma}{R(e)}, \mbmu{R}{\sigma}{R(e)}$ such that $\Pr_{(\eta,\eta')\sim\Psi_{R(e)}(R,\sigma,\tau)}[\eta\mid_{R(e)}=\eta'\mid_{R(e)}]$ is maximized, i.e.~the discrepancy $\Disc{\Psi_{R(e)}(R,\sigma,\tau)}{R(e)}$ is minimized. We first give an upper bound on $\Disc{\Psi_{R(e)}(R,\sigma,\tau)}{R(e)}$.

Let $\eta\in\{0,1\}^R$ be a configuration on $R$, we use $n_\eta$ to denote the number of edges in $R(e)$ that are assigned to $1$ by $\eta$. Let 
\begin{align*}
w_e=\sum_{\substack{\eta\in\{0,1\}^R\\n_\eta\mbox{\scriptsize is even}}}w_R^\sigma(\eta),
\qquad
w_o=\sum_{\substack{\eta\in\{0,1\}^R\\n_\eta\mbox{\scriptsize is odd}}}w_R^\sigma(\eta),
\qquad
w_\eta=w^\sigma_R(\eta).
\end{align*}
Then for every $\eta\in\{0,1\}^R$, we have
$\bmu{R}{\sigma}(\eta)=\frac{w_\eta}{w_e+w_o}$ and 
\[
\bmu{R}{\tau}(\eta)=\left\{
  \begin{array}{ll}
    \frac{w_\eta\cdot \mu}{w_e\cdot\mu+w_o/\mu}, &\mbox{ if $n_\eta$ is even}\\
    \frac{w_\eta/\mu}{w_e\cdot\mu+w_o/\mu}, &\mbox{ if $n_\eta$ is odd}
  \end{array}\right.
\]
It holds that
\begin{align*}
&\mbox{if }n_\eta\mbox{ is even,}
&\bmu{R}{\sigma}(\eta)-\bmu{R}{\tau}(\eta)=\frac{w_\eta}{w_e+w_o}-\frac{w_\eta\cdot \mu}{w_e\cdot\mu+w_o/\mu}=w_\eta\cdot\frac{w_o(1-\mu^2)}{(w_e+w_o)(w_e\cdot\mu^2+w_o)}>0;\\
&\mbox{if }n_\eta\mbox{ is odd,}
&\bmu{R}{\sigma}(\eta)-\bmu{R}{\tau}(\eta)=\frac{w_\eta}{w_e+w_o}-\frac{w_\eta/\mu}{w_e\cdot\mu+w_o/\mu}=w_\eta\cdot\frac{w_e(\mu^2-1)}{(w_e+w_o)(w_e\cdot\mu^2+w_o)}<0.
\end{align*}
Thus in the coupling $\Psi_{R(e)}(R,\sigma,\tau)$, we have
 \begin{align}
   \Disc{\Psi_{R(e)}(R,\sigma,\tau)}{R(e)}
   &=\sum_{\eta:\mbox{ $n_\eta$ is even}} \bmu{R}{\sigma}(\eta)-\bmu{R}{\tau}(\eta)
   =\sum_{\eta:\mbox{ $n_\eta$ is even}}w_\eta\cdot\frac{w_o(1-\mu^2)}{(w_e+w_o)(w_e\cdot\mu^2+w_o)}\notag\\
   &=\frac{w_ew_o(1-\mu^2)}{(w_e+w_o)(w_e\cdot\mu^2+w_o)}
   =\frac{1-\mu^2}{(1+w_o/w_e)(1+\mu^2\cdot w_e/w_o)}.\label{eq:subgraph-disc}
 \end{align}
We then show that $\lambda \le \frac{w_e}{w_o}\le\frac{1}{\lambda\mu^2}$.
We assume a total order on all edges. For any $\eta$ with even $n_\eta$, let $\phi(\eta)$ be the configuration resulting from flipping the state of $\eta$ on the first edge in $R(e)$. Note that $\phi$ is a bijection between configurations in $\{0,1\}^R$ with even $n_\eta$ and those with odd $n_\eta$. 
It is easy to check that 
\begin{align*}
\mbox{if }n_\eta\mbox{ is even,}
\quad w_\eta\ge \lambda w_{\phi(\eta)};
\quad\mbox{and }
\mbox{if }n_\eta\mbox{ is odd,}
\quad w_{\eta}\ge \lambda\mu^2\cdot w_{\phi^{-1}(\eta)}.
\end{align*}
Combining with the fact that $\phi$ is a bijection, we prove that $\lambda \le \frac{w_e}{w_o}\le\frac{1}{\lambda\mu^2}$. Substituting this into~\eqref{eq:subgraph-disc}, we have $\Disc{\Psi_{R(e)}(R,\sigma,\tau)}{R(e)}\le\frac{1-\mu^2}{(1+\lambda\mu^2)^2}$. And since $\mu,\lambda>0$, all boundary configurations are $R$-feasible, thus for any boundary configurations $\sigma',\tau'$ differing on $t$ edges, we have $T(\sigma',\tau')= t$, i.e.~we can migrate from one boundary configuration to another by modifying one edge at a time without violating the feasibility during the process.
Therefore, 
if $\Delta<\frac{(1+\lambda\mu^2)^2}{1-\mu^2}$, then 
\[
\Disc{\Psi_{R(e)}(R,\sigma,\tau)}{R(e)}\cdot\max_{\substack{\sigma_1,\sigma_2\in[q]^{\ebond (R\setminus R(e))}\\\text{differ on }R(e)}}T(\sigma_1,\sigma_2)\le \frac{1-\mu^2}{(1+\lambda\mu^2)^2}\cdot|R(e)|<\frac{1}{\Delta}\cdot \Delta=1.
\]
For any $e_0\in R$, we can apply the recursion in Lemma \ref{lem:holantcouple} for $\mathrm{dist}(e,e_0)$ many times where $\mathrm{dist}(e,e_0)$ denotes the distance between $e$ and $e_0$, thus $\Disc{\Psi(R,\sigma,\tau)}{e_0}=\exp(\Omega(-\mathrm{dist}(e,e_0)))$. Then applying Lemma \ref{lemma-holant-ssssm}, since $T(\sigma',\tau')=\mathrm{poly}(n)$ for any $\sigma',\tau'$, we have the strong spatial mixing.
 \end{proof}

The ferromagnetic Ising model is a spin system specified by $\Phi_E:\{0,1\}^2\to\mathbb{R}^+$ and $\Phi_V:\{0,1\}\to\mathbb{R}^+$ such that $\Phi_E(x,y)=a>1$ if $x=y$ and $\Phi_E(x,y)=1$ if otherwise;,and $\Phi_V(x)=b>0$ if $x=1$ and $\Phi_V(x)=1$ if $x=0$. We call $(a,b)$ the parameters of the system. 

The Ising model can also be specified by the inverse temperature $\beta$ and external field $B$ as follows. Given a graph $G=(V,E)$, the partition function is defined as
\[
Z_{\mathrm{Ising}}(G)=\sum_{\sigma\in\{-1,1\}^V}\prod_{ij\in E}\exp(\beta\sigma(i)\sigma(j))\prod_{i\in V}\exp(B\sigma(i))
\]

\begin{theorem}[Jerrum and Sinclair~\cite{jerrum1993polynomial}]
Let $G(V,E)$ be a graph. Let $\lambda=\frac{a-1}{a+1}=\tanh\beta,\mu=|\frac{b-1}{b+1}|=\tanh B$, then
\[
  Z_\mathrm{Ising}(G)=M_G\cdot Z_{\mathrm{sub}}(G),
\]
for some $M_G$ which can be computed in polynomial time.
\end{theorem}
The transformation from the ferromagnetic Ising model to the subgraphs world model is actually a holographic transformation and the above theorem can be seen as a special case of Valiant's Holant theorem~\cite{valiant2008holographic,cai2011computational}.

Translating the conditions in Theorem \ref{thm-subgraph-world-decay} for subgraphs world back to the ferromagnetic Ising model, we have that $\Delta<\frac{(ab^2+a+2b)^2}{b(a+1)^2(b+1)^2}$ or equivalently $\Delta<\frac{(e^{2\beta+4B}+e^{2\beta}+2e^{2B})^2}{e^{2B}(e^{2\beta}+1)^2(e^{2B}+1)^2}$.

\subsection{Recursive coupling on spin systems}
Let $G(V,E)$ be an undirected graph, and $\Phi:[q]^2\to\mathbb{R}^+$ be a symmetric function of nonnegative values.
Consider the $q$-state spin system whose partition function is defined by
\begin{align*}
Z(G)
&=\sum_{\sigma\in[q]^V}w(\sigma)
\quad\mbox{ where }w(\sigma)=\prod_{\{u,v\}\in E}\Phi(\sigma(u),\sigma(v)).
\end{align*}
With this definition of weight, for any feasible $\sigma\in[q]^\Lambda$ for $\Lambda\subseteq V$, we can accordingly define the Gibbs measure $\mu^\sigma$ over $[q]^V$ and the marginal distribution at $\mu^\sigma_v$ at vertex $v$. 

\begin{definition}[Strong Spatial Mixing on Spin Systems]
 A spin system on a family of graphs has strong spatial mixing (SSM) if for any graph $G=(V,E)$ in the family, any $v\in V, \Lambda\subseteq V$ and any two feasible configurations $\sigma_\Lambda,\tau_\Lambda\in[q]^\Lambda$,
  \[
  \|\mu^{\sigma_\Lambda}_v-\mu^{\tau_\Lambda}_v\|_{\mathrm{TV}}\le\exp(-\Omega(\mathrm{dist}(v,\Delta))),
  \]
where $\Delta\subseteq\Lambda$ is the subset on which $\sigma_\Lambda$ and $\tau_\Lambda$ differ, and $\mathrm{dist}(v,S)$ is the shortest distance from $v$ to any vertex in $\Delta$.
\end{definition}
It is easy to verify that under this definition SSM of a spin system is equivalent to the SSM of the Holant problem on the bipartite incident graph $\mathcal{I}_G$ which simulates the original spin system. Thus to use the FPTAS which rely on the SSM for Holant problems, it is sufficient to prove SSM on the original spin system.


Let $R\subseteq V$, called a \concept{region}. Let $\vbond R=\{v\in V\setminus R\mid\exists uv\in E,u\in R\}$ be the vertex boundary of $R$ and $\ebond R=\{uv\in E\mid u\in R, v\not\in R\}$ be the edge boundary of $R$. Denote by $E(R,R)=\{uv\in E\mid u\in R,v\in R\}$ the set of internal edges in $R$.

An edge boundary configuration of $R$, or just \concept{boundary configuration} for short, is a $\sigma\in[q]^{\ebond R}$. For every boundary configuration $\sigma\in[q]^{\ebond R}$ and a configuration $\eta\in[q]^R$ of the region $R$, define the \concept{regional weight} as
\begin{align*}
w_R^\sigma(\eta)
&=
\prod_{uv\in E(R,R)}\Phi(\eta(u),\eta(v))\cdot
\prod_{\substack{e=uv\in\ebond R\\u\in R}}\Phi(\sigma(e),\eta(u)).
\end{align*}


We say that an edge boundary configuration $\sigma\in[q]^{\ebond R}$ is \concept{$R$-feasible} if there exists an $\eta\in[q]^R$ such that $w_R^\sigma(\eta)>0$. 
For $R$-feasible boundary configuration $\sigma\in[q]^{\ebond R}$, a \concept{regional Gibbs measure} $\bmu{R}{\sigma}$ over $[q]^R$ can be defined as that $\bmu{R}{\sigma}(\eta)=\frac{w_R^\sigma(\eta)}{\sum_{\pi\in[q]^R}w_R^\sigma(\pi)}$ for each $\eta\in[q]^R$.
For $v\in R$, let $\mbmu{R}{\sigma}{v}$ be the marginal distribution of $\bmu{R}{\sigma}$ at vertex $v$.


\begin{definition}
Let $R\subseteq V$, $v\in R$, and $\sigma,\tau\in[q]^{\ebond R}$ be two $R$-feasible boundary configurations. 
Let $\Psi(R,\sigma,\tau)$ be a coupling of $\bmu{R}{\sigma},\bmu{R}{\tau}$. Define the \concept{discrepancy} of $\Psi(R,\sigma,\tau)$ at vertex $v$ as
\[
\Disc{\Psi(R,\sigma,\tau)}{v}=\Pr_{(\eta,\eta')\sim\Psi(R,\sigma,\tau)}\left[\eta(v)\ne \eta'(v)\right].
\]
\end{definition}


\begin{lemma}\label{prop:bondtran}
Let $\Lambda\subset V$ and $\sigma,\tau\in[q]^{\Lambda}$ be two feasible vertex configurations differing on $\Delta\subseteq\Lambda$. Let $R=V\setminus\Lambda$, $v\in R$ and $\nabla=\{uv\in E\mid v\in R, u\in\Delta\}$ be the set of boundary edges of $R$ incident to vertices in $\Delta$.
There exist two $R$-feasible edge boundary configurations $\sigma',\tau'\in[q]^{\ebond R}$ differing only on edges in $\nabla$ such that
\begin{align*}
\left\|\mu_v^\sigma-\mu_v^\tau\right\|_{\mathrm{TV}}
&=
\left\|\mbmu{R}{\sigma'}{v}-\mbmu{R}{\tau'}{v}\right\|_{\mathrm{TV}}
\le
T(\sigma',\tau')
\cdot
\max_{\substack{\sigma_1,\sigma_2\in[q]^{\ebond R}\\\text{ differ on an }e\in\nabla}}\|\mbmu{R}{\sigma_1}{v}-\mbmu{R}{\sigma_2}{v}\|_{\mathrm{TV}}\\
&\le
T(\sigma',\tau')
\cdot
\max_{\substack{\sigma_1,\sigma_2\in[q]^{\ebond R}\\\text{ differ on an }e\in\nabla}}\Disc{\Psi(R,\sigma_1,\sigma_2)}{v},
\end{align*}
for arbitrary coupling $\Psi(R,\sigma_1,\sigma_2)$ of $\bmu{R}{\sigma_1},\bmu{R}{\sigma_2}$.
\end{lemma}
\begin{proof}
Let $\sigma',\tau'\in[q]^{\ebond R}$ be defined that the assigned value of each boundary edge is consistent with that of the incident boundary vertex in $\Lambda$ assigned by $\sigma,\tau$ respectively. The rest is the same as proof of Lemma \ref{lemma-holant-ssssm}.
\end{proof}
We describe the recursive coupling for spin systems introduced in~\cite{goldberg2005strong}.

\paragraph{The recursive coupling.} Let $R\subseteq V$ and $v_0\in R$. Let $\sigma,\tau\in[q]^{\ebond R}$ be any two $R$-feasible boundary configurations that differ on only one edge $uv\in\ebond R$ where $v\in R$ and $u\not\in R$.  Let $E_R(v)=\{wv\in E\mid w\in R\}$ be set of internal edges of $R$ incident to $v$.
Let $\Psi_v(R,\sigma,\tau)$ be a coupling of marginal distributions $\mbmu{R}{\sigma}{v},\mbmu{R}{\tau}{v}$ at vertex $v$. A coupling $\Psi(R,\sigma,\tau)$ of $\bmu{R}{\sigma},\bmu{R}{\tau}$ can be recursively constructed by the coupling rule $\Psi_v(R,\sigma,\tau)$ at the vertices $v$ incident to the only disagreeing edge: 

Let $(\eta,\eta')\in[q]^R\times[q]^R$ denote the pair of configurations sampled from $\Psi(R,\sigma,\tau)$.
\begin{enumerate}
\item (Base case) If $v=v_0$, sample $(\eta(v),\eta'(v))$ according to $\Psi_v(R,\sigma,\tau)$ and arbitrarily sample the rest of $(\eta,\eta')$  conditioning on $(\eta(v),\eta'(v))$ as long as $(\eta,\eta')$ is a faithful coupling of $\bmu{R}{\sigma},\bmu{R}{\tau}$.
If $v\ne v_0$ and $|E_R(v)|=0$, in which case $v$ and $v_0$ are disconnected in $G[R]$, sample $(\eta,\eta')$ such that $(\eta(v_0),\eta'(v_0))$ is perfectly coupled.
\item
(General case)
$v\ne v_0$ and $v$ has neighbors in $R$, i.e.~$|E_R(v)|>0$. 
Sample $(\eta(v),\eta'(v))$ according to $\Psi_v(R,\sigma,\tau)$. Denote $(x,y)=(\eta(v),\eta'(v))$. 
Construct new region and boundaries as: $R'=R\setminus\{v\}$; $\recolor{\sigma}{}{x}{}\in[q]^{\ebond R'}$ agrees with $\sigma$ on common edges and assigns $x$ to edges in $E_R(v)$; and $\recolor{\tau}{}{x}{y}\in[q]^{\ebond R'}$ agrees with $\tau$ on common edges and assigns $y$ to edges in $E_R(v)$. 
The rest of $(\eta,\eta')$ is sampled from a coupling $\Psi(x,y)$ of $\bmu{R'}{\recolor{\sigma}{}{x}},\bmu{R'}{\recolor{\tau}{}{y}}$. If $x=y$, $\recolor{\sigma}{}{x}=\recolor{\tau}{}{y}$ and $\Psi(x,y)$ is a perfect coupling. 
If $x\neq y$, let 
$\recolor{\sigma}{1}{(x,y)},\ldots,\recolor{\sigma}{t}{(x,y)}$ be a feasible path from $\recolor{\sigma}{}{x}$ to $\recolor{\tau}{}{y}$ of length $t=T\left(\recolor{\sigma}{}{x},\recolor{\tau}{}{y}\right)$.
Let $\Psi(x,y)$ be the composition of coupling $\Psi\left(R',\recolor{\sigma}{i}{(x,y)},\recolor{\sigma}{i+1}{(x,y)}\right), i=1,2,\ldots,t-1$, in the same manner as path coupling, where each $\Psi\left(R',\recolor{\sigma}{i}{(x,y)},\recolor{\sigma}{i+1}{(x,y)}\right)$ is recursively defined. It is easy to verify that $\Psi(x,y)$ is a coupling of $\bmu{R'}{\recolor{\sigma}{}{x}},\bmu{R'}{\recolor{\tau}{}{y}}$.
This complete the construction of $\Psi(R,\sigma,\tau)$.
\end{enumerate}

The following lemma is proved in~\cite{goldberg2005strong} for the recursive coupling constructed as above.
\begin{lemma}[Goldberg-Martin-Paterson~\cite{goldberg2005strong}]\label{lem:spincouple}
\begin{align*}
\Disc{\Psi(R,\sigma,\tau)}{v_0}
&\le
\sum_{\substack{x,y\in[q]\\x\neq y}}\Pr_{(\eta,\eta')\sim\Psi_v(R,\sigma,\tau)}[\eta(v)=x\wedge \eta'(v)=y]
\cdot\sum_{i=1}^{T\left(\recolor{\sigma}{}{x},\recolor{\tau}{}{y}\right)-1}\Disc{\Psi\left(R\setminus\{v\},\recolor{\sigma}{i}{(x,y)},\recolor{\sigma}{i+1}{(x,y)}\right)}{v_0}\\
&\le
\Disc{\Psi_v(R,\sigma,\tau)}{v}
\cdot\max_{\substack{\sigma_1,\sigma_2\in[q]^{\ebond(R\setminus\{v\})}\\\text{differ on }E_R(v)}}T\left(\sigma_1,\sigma_2\right)
\cdot\max_{\substack{\sigma_1,\sigma_2\in[q]^{\ebond (R\setminus\{v\})}\\\text{differ on an }e\in E_R(v)}}
\Disc{\Psi(R\setminus\{v\},\sigma_1,\sigma_2)}{v_0}.
\end{align*}
\end{lemma}

\subsection{Ferromagnetic Potts Model}
The Potts model is that a $q$-state spin system defined as that $\Phi(x,y)=\lambda$ if $x=y$ and $\Phi(x,y)=1$ if otherwise. We also write $\lambda=e^\beta$ where $\beta$ is the \concept{inverse temperature}. 
The weight of a configuration $\sigma\in[q]^E$ is $w(\sigma)=\lambda^{\mathrm{mon}_\sigma(E)}$, where 
where $\mathrm{mon}_\sigma(E)$ is the number of monochromatic edges, i.e.~edges $uv\in E$ that $\sigma(u)=\sigma(v)$.
A Potts model is \concept{ferromagnetic} if $\beta>0$ or equivalently if $\lambda>1$.

\begin{theorem}
Let $\mathcal{G}$ be any family of graphs whose maximum degree is bounded by $\Delta$.
A ferromagnetic Potts model on graph family $\mathcal{G}$ has strong spatial mixing if $q-2>\left(\lambda-1\right)\left(\Delta-1\right)\lambda^\Delta$, or in terms of inverse temperature if $\beta<\frac{\ln\left(\frac{q-2}{\Delta-1}\right)}{\Delta+1}$.
\end{theorem}

\begin{proof}
Let $R\subseteq V$ be a region and $\sigma,\tau\in[q]^{\ebond R}$ be two $R$-feasible boundary configurations differing on $e=uv\in\ebond R$ with $v\in R$ and $u\not\in R$. The regional weights $w_R^{\sigma},w_R^{\tau}$ and regional Gibbs measures $\bmu{R}{\sigma},\bmu{R}{\tau}$ can be defined accordingly.

Let $\Psi_v(R,\sigma,\tau)$ be the coupling of $\mbmu{R}{\sigma}{v},\mbmu{R}{\tau}{v}$ that $\Pr_{(\eta,\eta')\sim\Psi_v(R,\sigma,\tau)}[\eta(v)=\eta'(v)]$ is maximized, i.e.~the discrepancy $\Disc{\Psi_v(R,\sigma,\tau)}{v}$ at $v$ is minimized. We first give an upper bound on $\Disc{\Psi_v(R,\sigma,\tau)}{v}$.

Consider a boundary configuration $\sigma'$ that agrees with $\sigma$ on $\ebond R\setminus\{e\}$. Let $\sigma'(e)=q$, a free color not in $[q]$, and override the definition of function $\Phi$ such that $\Phi(i,q)=\Phi(q,i)=1$ for all $i\in[q]$. 
We let 
$c_i=\sum_{\substack{\eta\in[q]^R\\\eta(v)=i}}w_R^{\sigma'}(\eta)$.
Without loss of generality, assume $\sigma(e)=0, \tau(e)=1$ and $c_0\ge c_1$. Denote that $c=\sum_{i=2}^{q-1}c_i$. We have
\begin{align*}
  \sum_{\substack{\eta\in[q]^R\\\eta(v)=i}}\bmu{R}{\sigma}(\eta)&=
  \left\{\begin{array}{ll}
      \frac{\lambda c_0}{\lambda c_0+c_1+c}, & i=0;\\
      \frac{c_i}{\lambda c_0+c_1+c}, & 1\le i< q;
    \end{array}\right.
\quad\mbox{and}\quad
  \sum_{\substack{\eta\in[q]^R\\\eta(v)=i}}\bmu{R}{\tau}(\eta)&=
  \left\{
    \begin{array}{ll}
      \frac{\lambda c_1}{c_0+\lambda c_1+c}, & i=1;\\
      \frac{c_i}{c_0+\lambda c_1+c}, & i=0,2\le i< q.
    \end{array}\right.
\end{align*} 
For ferromagnetic Potts model,  $\lambda > 1$ and $\lambda c_0+c_1+c>c_0+\lambda c_1+c$, thus 
$\sum_{\substack{\eta\in[q]^R\\\eta(v)=i}}\bmu{R}{\sigma}(\eta)<\sum_{\substack{\eta\in[q]^R\\\eta(v)=i}}\bmu{R}{\tau}(\eta)$
for all $i \ne 0$. Therefore in the coupling $\Psi_v(R,\sigma,\tau)$, we have
\begin{align*}
  \Disc{\Phi_v(R,\sigma,\tau)}{v}
  &=\sum_{\substack{\eta\in[q]^R\\\eta(v)=0}}\bmu{R}{\sigma}(\eta)-\sum_{\substack{\eta\in[q]^R\\\eta(v)=0}}\bmu{R}{\tau}(\eta)
  =\frac{\lambda c_0}{\lambda c_0+c_1+c}-\frac{c_0}{c_0+\lambda c_1+c}\\
  &=\frac{c_0(\lambda-1)(c_1+\lambda c_1+c)}{(\lambda c_0+c_1+c)(c_0+\lambda c_1+c)}
  \le\frac{(\lambda-1)c_0}{c_0+\lambda c_0+c}.\\
\end{align*}
The value of the last term  $\frac{(\lambda-1)c_0}{c_0+\lambda c_0+c}$ depends on $R,\sigma$ and $e$. For fixed $e$, define $\nu(R,\sigma)=\frac{(\lambda-1)c_0}{c_0+\lambda c_0+c}$, where $c_0,c$ are defined by $R,\sigma$, and $e$ as above.

Let $R'\subseteq R$ be any region containing $v$. Let $\pi\in[q]^{\ebond R'}$ be an $R'$-feasible boundary configuration that agrees with $\sigma$ on common edges and maximizes $\nu(R',\pi)$. The next lemma is very similar to the one proved in~\cite{goldberg2006improved}.
\begin{claim}
  $\nu(R,\sigma)\le\nu(R',\pi)$.
\end{claim}
\begin{proof}
Let $\eta$ be an $R'$-feasible boundary configuration of $R'$ such that $\eta(e)=q$ the free color and $\eta$ agrees with $\sigma$ on all other common edges. We use $c_i^\eta$ to denote the sum of weight of configurations in $R'$ which assign spin $i$ to $v$ with boundary configuration $\eta$ and $c^\eta=\sum_{i=2}^{q-1}c_i^\eta$. Recall that $\nu(R,\sigma) =\frac{\lambda-1}{\lambda+1+c/c_0}$.
Then the claim follows from the fact that $c/c_0$ is a convex combination of $c^\eta/c_0^\eta$ over all such $\eta$.
\end{proof}

Consider $R'=\{v\}$, suppose $v$ has $k$ neighbours in $R\cup\vbond R$, where $0\le k<\Delta$ ($u$ is not taken into account as it is always assigned free color in the definition of $c_i$s). We want to find a boundary condition $\pi$ that maximize $\nu(R',\pi)$, which is equivalent to minimize $c/c_0$. 

Since $\lambda>1$, the ratio achieves its minimum when $\pi$ assigns all the $k$ incident edges of $v$ with spin $0$. In this case, $c_0=\lambda^k, c=\sum_{i=2}^{q-1}c_i=q-2$. Thus assuming that $q-2>\left(\lambda-1\right)\left(\Delta-1\right)\lambda^\Delta$, we have
\begin{align*}
  \Disc{\Phi_v(R,\sigma,\tau)}{v}
  &\le\nu(R,\sigma)
  \le\nu(R',\pi)
  \le \frac{(\lambda-1)\lambda^k}{q-2+(\lambda+1)\lambda^k}
  \le \frac{(\lambda-1)\lambda^\Delta}{q-2+(\lambda+1)\lambda^\Delta}
  < \frac{1}{\Delta}.
\end{align*}
Note that for the ferromagnetic Potts model, all configurations are feasible, and $|E_R(v)|\le\Delta$, thus 
\begin{align*}
  \max_{\substack{\sigma_1,\sigma_2\in[q]^{\ebond(R\setminus\{v\})}\\\text{differ on }E_R(v)}}T\left(\sigma_1,\sigma_2\right)=|E_R(v)|\le \Delta.
\end{align*}
Applying Lemma \ref{lem:spincouple} and Lemma  \ref{prop:bondtran}, we have the strong spatial mixing.


\end{proof}

\subsection{Algorithmic implications}
Both subgraphs world and ferromagnetic Potts model are Holant problems of regular constraint functions and both satisfy tractable search. Then by Theorem \ref{thm-FPTAS}, 
we have the following algorithmic results.
}
\begin{theorem}
Let $\mathcal{G}$ be the family of apex-minor-free graphs of maximum degree $\Delta$.
\begin{itemize}
\item
If $\Delta<\frac{(1+\lambda\mu^2)^2}{1-\mu^2}$,
there exists an FPTAS for subgraphs world of parameters $0<\mu,\lambda<1$ on graphs from $\mathcal{G}$.
\item
If $\Delta<\frac{(e^{2\beta+4B}+e^{2\beta}+2e^{2B})^2}{e^{2B}(e^{2\beta}+1)^2(e^{2B}+1)^2}$,
there exists an FPTAS for ferromagnetic Ising model of inverse temperature $\beta$ and external filed $B$ on graphs from $\mathcal{G}$.
\item
If $\beta<\frac{\ln\left(\frac{q-2}{\Delta-1}\right)}{\Delta+1}$,
there exists an FPTAS for $q$-state ferromagnetic Potts model of inverse temperature $\beta$ on graphs from $\mathcal{G}$.
\end{itemize}
\end{theorem}
\ifabs{
The analysis of correlation decay and the proof of the theorem can be found in the full version of the paper in Appendix.
}{
The FPTAS for Ising model is not by applying Theorem \ref{thm-FPTAS} but due to the transformation between subgraphs world and Ising model.}

\subsubsection*{Acknowledgement.}
We would like to thank Jin-Yi Cai, Heng Guo, and Pinyan Lu for the in-depth discussions. Thank Alistair Sinclair and Leslie Valiant for their comments and interests.

\ifarxiv{

}{
\bibliographystyle{abbrv}
\bibliography{refs}
}


\end{document}